%% file: main.tex
\documentclass[11pt]{article}

\usepackage[a4paper,margin=1in]{geometry}
\usepackage{amsmath,amssymb,amsfonts,amsthm,mathtools,bm,mathrsfs}
\usepackage{booktabs,array,enumitem,microtype,authblk}
\usepackage{xcolor}
\usepackage[colorlinks=true,linkcolor=blue,citecolor=blue,urlcolor=blue]{hyperref}
\hypersetup{
  pdftitle={Geometric Completion of Thermodynamic Response: Recognition Kernels, Cut-Flow Jets, Constructive Continuum Sectors, and Action-Lifted Spacetime Field Equations},
  pdfauthor={Monty Dabas},
  pdfsubject={Thermodynamic response geometry, Recognition Kernel theory, cut-flow jets, constructive continuum measure, and action-lifted field equations},
  pdfkeywords={thermodynamics, recognition kernel, cut-flow jets, continuum measure, reflection positivity, gauge fields, gravity, effective action}
}

\newtheorem{definition}{Definition}[section]
\newtheorem{theorem}[definition]{Theorem}
\newtheorem{proposition}[definition]{Proposition}
\newtheorem{lemma}[definition]{Lemma}
\newtheorem{corollary}[definition]{Corollary}
\newtheorem{remark}[definition]{Remark}
\newtheorem{principle}[definition]{Principle}

\newcommand{\R}{\mathbb R}
\newcommand{\C}{\mathbb C}
\newcommand{\N}{\mathbb N}
\newcommand{\dd}{\mathrm d}
\newcommand{\id}{\mathrm I}
\newcommand{\Ker}{\operatorname{Ker}}
\newcommand{\Ran}{\operatorname{Ran}}
\newcommand{\rank}{\operatorname{rank}}

\newcommand{\Sym}{\operatorname{Sym}}

\newcommand{\tr}{\operatorname{tr}}
\newcommand{\e}{\mathrm e}
\newcommand{\RK}{\mathrm{RK}}

\newcommand{\Eye}{\mathrm{Eye}}
\newcommand{\cH}{\mathcal H}
\newcommand{\cM}{\mathcal M}
\newcommand{\cJ}{\mathcal J}
\newcommand{\cA}{\mathcal A}
\newcommand{\cK}{\mathcal K}
\newcommand{\cB}{\mathfrak B}

\newcommand{\cR}{\mathcal R}

\newcommand{\cO}{\mathcal O}
\newcommand{\cN}{\mathcal N}
\newcommand{\cY}{\mathcal Y}
\newcommand{\cZ}{\mathcal Z}
\newcommand{\abs}[1]{\left|#1\right|}
\newcommand{\norm}[1]{\left\lVert#1\right\rVert}
\newcommand{\ip}[2]{\left\langle #1,#2\right\rangle}

\newcommand{\pdc}[3]{\left(\frac{\partial #1}{\partial #2}\right)_{#3}}

\title{Geometric Completion of Thermodynamic Response:\\
From the \texorpdfstring{$T$--$V$--$S$--$P$}{T-V-S-P} Compass to Recognition Kernels, Cut--Flow Jets, and Action-Lifted Spacetime Field Equations}
\author{Monty Dabas\thanks{ORCID: \href{https://orcid.org/0009-0005-6948-209X}{0009-0005-6948-209X}}}
\date{Rigorous revision draft, 4 August 2026}

\begin{document}
\maketitle

\begin{abstract}
We give a theorem-level reconstruction of the geometric thermodynamic response programme for a simple compressible system.  The equilibrium object is a two-dimensional Legendre submanifold generated by a smooth fundamental relation $U(S,V)$; the familiar $T$--$V$--$S$--$P$ compass records its four regular Legendre charts.  From the exact potential differentials we derive six constraint-dependent response coordinates, their conjugate product identities, and the equilibrium closure
\[
\frac{\lambda_P}{\lambda_V}\frac{\lambda_S}{z_T}=1,
\]
which is equivalent to the classical identity $C_P/C_V=K_S/K_T$.  The intrinsic antisymmetric structure is the exterior two-form and its Jacobian bracket.  For any four nondimensional response channels, their pairwise area brackets form a canonical skew $4\times4$ matrix of rank at most two satisfying an exact Pluecker relation; this replaces, rather than assumes, a skew ordinary derivative Jacobian.

We introduce a Recognition Kernel theorem.  For an observation map $\cA$ and target $\Pi$, the target-relevant blind quotient $\Ker\cA/(\Ker\cA\cap\Ker\Pi)$ determines the canonical target-faithful completion and, in finite dimensions, the exact minimum number of supplemental observations.  A sharp decoder-burden theorem characterizes when a target functional factors through the observation.  For cyclic response one-forms, Stokes' theorem gives the exact open residue after a prior lawful ledger is removed.  A positive constitutive calibration of open curvature then yields a rigorous nonnegative production law, while the calibration data remain distinct from geometry itself.

A Cut--Flow--Jet Recognition theorem proves that a nilpotent covariant raiser generates every finite response-jet layer, that a self-adjoint involutive cut gives unique cut and adjoint sectors, and that the first cut-loop correction is the commutator curvature.  The finite construction extends to a closed unbounded infinite-jet raiser on analytic response seeds, with exact convergence and uniqueness of the analytic jet flow.  Holomorphic constitutive laws along analytic process curves are shown to satisfy the analytic-seed estimate away from response singularities.  These results are assembled in an Extended Thermodynamic Recognition Capstone.  A noise-stable Pluecker estimate gives a finite-error falsification test for the four-response geometry.

A local-equilibrium state field $\Theta:M\to\cM_{\mathrm{th}}$ pulls the thermodynamic response-jet bundle to spacetime.  Smooth constant-rank observation and target maps then produce the target-relevant quotient bundle $E_{\RK}=\Ker\cA/(\Ker\cA\cap\Ker\Pi)$ and the explicit Recognition field $\Phi_{\mathrm{th}}=qP_{\Ker\cA}\Theta^*(j^N\Lambda)$.  Compatible jet transport descends to a connection on $E_{\RK}$; its parallel transports form an exact target-faithful commuting square, and its curvature is the infinitesimal failure of path-independent thermodynamic reconstruction.  A scalar principal symbol for target-faithful quotient propagation reconstructs the Lorentzian co-metric by polarization and prevents a target-relevant characteristic cone from remaining blind after completion.  A uniform decoder theorem proves source domination and strict target no-blindness.  Within a declared parity-even local low-derivative class, an invariant-action classification gives the Einstein--Hilbert plus gauge--matter density used below, up to couplings, an invariant potential, and boundary terms.

Finally, using this constructed spacetime Recognition bundle and its minimal invariant action, we derive Einstein-type metric, Recognition gauge, and Recognition-field equations, together with Bianchi and Noether conservation laws and exact Einstein--Maxwell and Yang--Mills--matter reductions.  We prove that the kinematic capstone alone cannot fix the numerical couplings or nonlinear potential.  A coercive finite regulator yields exact Schwinger--Dyson and effective quantum equations.  A continuum quantum-effective gravity equation follows under the separately stated existence, differentiability, symmetry, and anomaly-free hypotheses for a renormalized effective action.  For every stationary background whose gauge-fixed target-faithful physical Hessian is positive and ultrastatic, we construct the regulator-independent Gaussian continuum measure, prove spectral regulator removal, Osterwalder--Schrader positivity, exact bosonic anomaly freedom, and a positive-Hamiltonian unitary Lorentzian reconstruction.  Null propagation plus one scale calibration and complete parallel transport identify the metric and connection up to diffeomorphism and gauge.  A fixed-decoder theorem further converts lower linearized observation bounds and derivative stability into explicit finite-error reconstruction estimates for noisy spacetime data.  The nonlinear interacting ultraviolet completion and the empirical verification of the identification hypotheses remain separate obligations.
\end{abstract}


\tableofcontents

\input{sections/00_dependency_contract.tex}
\input{sections/01_foundations.tex}
\input{sections/01a_response_summary.tex}
\input{sections/02_recognition.tex}
\input{sections/02a_measurement_operators.tex}
\input{sections/03_loop_and_cutflow.tex}
\input{sections/03a_thermodynamic_cut_realizations.tex}
\input{sections/04_capstone_and_examples.tex}
\input{sections/05_proved_extensions.tex}
\input{sections/05a_curvature_and_analytic_relevance.tex}
\input{sections/05b_extended_capstone.tex}
\input{sections/05c_novelty_falsifiability.tex}
\input{sections/06a0_thermodynamic_recognition_square.tex}
\input{sections/06a00_response_propagation_metric.tex}
\input{sections/06a_spacetime_recognition_datum.tex}
\input{sections/06a1_dynamical_nonselection.tex}
\input{sections/06b_classical_field_equations.tex}
\input{sections/06c_quantum_effective_equations.tex}
\input{sections/06d_constructive_continuum_and_identification.tex}
\input{sections/06e_noise_stable_spacetime_recognition.tex}
\input{sections/05_conclusion_and_ledger.tex}

\input{sections/references.tex}
\end{document}

%% file: sections/00_dependency_contract.tex
\section{Reader contract and theorem dependency map}
\label{sec:reader-contract}

This paper is written primarily for readers in mathematical thermodynamics, functional analysis, inverse problems, differential geometry, and constructive field theory.  No reader is expected to treat the later spacetime sections as consequences of equilibrium notation alone.  The theorem chain is modular, and every arrow below names its additional datum.

\begin{center}
\small
\begin{tabular}{p{0.19\textwidth}p{0.34\textwidth}p{0.36\textwidth}}
\toprule
Module & Input added & Output proved\\
\midrule
Equilibrium & $U(S,V)$, stable chart & Legendre structure, Maxwell and response identities\\
Recognition & observation, target, ledger & blind quotient, minimal completion, decoder, residue\\
Response geometry & area form, four channels & skew rank-two bracket, Pluecker identity and noise test\\
Cut--Flow--Jet & connection, graded carrier, cut & finite/analytic jets, cut sectors, commutator curvature\\
Recognition square & $\Theta$, constant-rank bundle maps, compatible transport & quotient field bundle, explicit $\Phi_{\mathrm{th}}$, connection, curvature, source domination\\
Action lift & Lorentzian metric and minimal action contract & classified action form and residual constitutive freedom\\
Classical dynamics & couplings and invariant potential & metric, gauge, matter, Bianchi, and conservation equations\\
Finite quantum & coercive gauge-fixed regulator & exact Schwinger--Dyson and effective equations\\
Constructive sector & positive ultrastatic physical Hessian & continuum Gaussian measure, OS positivity, unitarity, bosonic anomaly freedom\\
Identification & propagation, scale, transport, decoder stability & exact and finite-error metric--connection reconstruction\\
\bottomrule
\end{tabular}
\end{center}

The central new interface is the \emph{thermodynamic Recognition square}:
\begin{equation}
\text{thermodynamic response jets}
\longrightarrow
E_{\RK}
\longrightarrow
(\Phi_{\mathrm{th}},A,F_A)
\longrightarrow
\text{classified action lift}.
\label{eq:opening-recognition-chain}
\end{equation}
The first two arrows are construction theorems.  The final arrow requires the stated locality, symmetry, parity, and derivative-order contract.  Numerical couplings and the nonlinear potential remain constitutive data.  The quantum-effective equations remain conditional on the existence of the stated effective action, while the positive quadratic sector is constructed independently and nonperturbatively.

The detailed novelty and falsifiability audit appears in Section~\ref{sec:dependency-falsifiability}.  Standard ingredients are identified there as standard.  The contribution claimed here is not the rediscovery of Einstein variation, Gaussian measures, or null-cone geometry; it is the target-relative Recognition architecture that connects them through explicit carriers, quotients, transport squares, decoder bounds, and claim boundaries.

%% file: sections/01_foundations.tex
\section{Purpose, revision scope, and claim discipline}

The purpose of this revision is not to enlarge the vocabulary of thermodynamics.  It is to identify the mathematical object actually present, prove the statements that follow from it, and separate those statements from constitutive or experimental hypotheses.  The equilibrium and response conventions follow standard thermodynamics \cite{Callen,deGrootMazur}; the geometric language is used only where its hypotheses are stated.  Throughout, complex Hilbert-space inner products are linear in the first argument.

The paper retains four ideas from the earlier geometric presentation \cite{DabasThermoV2}:
\begin{enumerate}[label=(\roman*)]
\item the signed $T$--$V$--$S$--$P$ compass as a compact display of the four Legendre charts;
\item constraint-dependent response coordinates denoted by $\lambda$ and $z$;
\item recursive response derivatives, now formulated as a finite covariant jet tower;
\item cyclic residue, now formulated through a Recognition Kernel and a prior lawful ledger.
\end{enumerate}

The revision makes four structural corrections.
\begin{enumerate}[label=(\roman*)]
\item A simple compressible equilibrium state space has two independent coordinates after the amount of substance is fixed.  The four symbols $T,V,S,P$ are overlapping chart variables, not four independent coordinates on an open subset of $\R^4$.
\item Maxwell relations follow from exact potential differentials, equivalently from the vanishing pullback of the thermodynamic contact form.  They do not require a universal skew Jacobian of the $\lambda$-map.
\item A nonzero loop integral is a geometric residue.  It becomes memory, hysteresis, or entropy production only after the relevant state extension, ledger, time orientation, units, and constitutive calibration are supplied.
\item The Cut--Flow--Jet construction is proved on a finite graded carrier.  Infinite completed jet spaces and unbounded generators require additional analytic hypotheses and are not silently inherited from the finite theorem.
\end{enumerate}

\begin{principle}[Claim boundary]
Every statement below is labelled by its hypotheses.  A constitutive example proves an identity inside that model; a finite certificate verifies a declared finite witness; neither is promoted into a universal physical law.
\end{principle}

\section{The equilibrium thermodynamic manifold}

\subsection{Admissible simple-compressible chart}

\begin{definition}[Admissible equilibrium chart]
Let $\Omega\subset\R^2$ be an open connected set with coordinates $(S,V)$ and let
\[
U\in C^{r+2}(\Omega),\qquad r\ge2,
\]
be a molar or specific fundamental relation.  Define
\[
T=\pdc{U}{S}{V},\qquad P=-\pdc{U}{V}{S}.
\]
The chart is \emph{admissible} when $T>0$ and the Hessian $D^2U$ is positive definite.  Thus
\[
U_{SS}>0,\qquad U_{VV}>0,\qquad
\det D^2U=U_{SS}U_{VV}-U_{SV}^2>0.
\]
Phase-transition singularities and chart-degeneracy loci are excluded from $\Omega$.
\end{definition}

Positive definiteness is stronger than is needed for several local identities, but it gives the standard stable single-phase regime and guarantees the regularity of the Legendre charts used below.

\subsection{Contact formulation and Maxwell relations}

Consider the five-dimensional contact space with coordinates $(U,S,V,T,P)$ and contact form (compare the standard contact-geometric thermodynamic formulation \cite{Mrugala1978})
\[
\vartheta=\dd U-T\,\dd S+P\,\dd V.
\]
The equilibrium embedding is
\[
\iota:\Omega\longrightarrow\R^5,
\qquad
\iota(S,V)=\bigl(U(S,V),S,V,U_S(S,V),-U_V(S,V)\bigr).
\]

\begin{theorem}[Equilibrium Legendre theorem]
For an admissible equilibrium chart,
\[
\iota^*\vartheta=0
\qquad\text{and}\qquad
\iota^*(\dd\vartheta)=0.
\]
Equivalently,
\begin{equation}
\dd T\wedge\dd S=\dd P\wedge\dd V
\label{eq:fundamental-two-form}
\end{equation}
on $\Omega$.  In every regular Legendre chart this identity is equivalent to the corresponding Maxwell relation.
\end{theorem}

\begin{proof}
By definition,
\[
\iota^*\vartheta
=\dd U-U_S\,\dd S-U_V\,\dd V=0.
\]
Exterior differentiation commutes with pullback, hence
\[
0=\dd(\iota^*\vartheta)=\iota^*(\dd\vartheta)
=-\dd T\wedge\dd S+\dd P\wedge\dd V,
\]
which gives \eqref{eq:fundamental-two-form}.  In the $(S,V)$ chart, the coefficient of $\dd S\wedge\dd V$ is
\[
-\pdc{T}{V}{S}-\pdc{P}{S}{V},
\]
so \eqref{eq:fundamental-two-form} gives
\[
\pdc{T}{V}{S}=-\pdc{P}{S}{V}.
\]
The other Maxwell relations follow after the regular Legendre coordinate changes proved below.
\end{proof}

\subsection{Four potential corners}

Define the Legendre potentials
\[
H=U+PV,\qquad F=U-TS,\qquad G=U+PV-TS.
\]
The positive-definite Hessian makes the maps to $(S,P)$, $(T,V)$, and $(T,P)$ locally invertible.

\begin{proposition}[Exact potential differentials]
On every admissible chart,
\begin{align}
\dd U&=T\,\dd S-P\,\dd V,\label{eq:dU}\\
\dd H&=T\,\dd S+V\,\dd P,\label{eq:dH}\\
\dd F&=-S\,\dd T-P\,\dd V,\label{eq:dF}\\
\dd G&=-S\,\dd T+V\,\dd P.\label{eq:dG}
\end{align}
Consequently,
\begin{align}
\pdc{T}{V}{S}&=-\pdc{P}{S}{V},&
\pdc{T}{P}{S}&=\pdc{V}{S}{P},\label{eq:maxwell1}\\
\pdc{S}{V}{T}&=\pdc{P}{T}{V},&
\pdc{S}{P}{T}&=-\pdc{V}{T}{P}.\label{eq:maxwell2}
\end{align}
\end{proposition}

\begin{proof}
The differential identities follow by differentiating the Legendre transforms and using \eqref{eq:dU}.  Equality of mixed partial derivatives of $U,H,F,G$ gives \eqref{eq:maxwell1}--\eqref{eq:maxwell2}.
\end{proof}

The compass is therefore a chart mnemonic for one Legendre submanifold.  It is useful, but the invariant object is the contact pullback and the exact differentials.

\section{Constraint-dependent response coordinates}

\subsection{Definitions}

Define the usual heat capacities
\[
C_P=T\pdc{S}{T}{P},
\qquad
C_V=T\pdc{S}{T}{V},
\]
and assume they are finite and nonzero on the chart under consideration.

The symbols below denote derivatives after the relevant potential is pulled back to the indicated regular chart.  This removes the common ambiguity of differentiating a potential with respect to a variable that is not one of its natural coordinates.

\begin{definition}[Six response coordinates]
Set
\begin{align}
\lambda_P&:=\pdc{G}{S}{P},&
\lambda_V&:=\pdc{F}{S}{V},\label{eq:lambda-caloric-def}\\
\lambda_S&:=\pdc{H}{V}{S},&
z_S&:=\pdc{U}{P}{S},\label{eq:lambda-adiabatic-def}\\
\lambda_T&:=\pdc{F}{P}{T},&
z_T&:=\pdc{G}{V}{T}.\label{eq:lambda-isothermal-def}
\end{align}
\end{definition}

The pairs $(\lambda_P,\lambda_V)$, $(\lambda_S,z_T)$, and $(z_S,\lambda_T)$ have, respectively, common physical dimensions.  They should not be assembled into a Euclidean vector without a declared nondimensionalization or direct-sum unit structure.

\subsection{Exact response formulas}

\begin{theorem}[Response-coordinate theorem]
On an admissible regular chart,
\begin{align}
\lambda_P&=-\frac{ST}{C_P},&
\lambda_V&=-\frac{ST}{C_V},\label{eq:lambda-caloric}\\
\lambda_S&=V\pdc{P}{V}{S},&
z_S&=-P\pdc{V}{P}{S},\label{eq:lambda-adiabatic}\\
\lambda_T&=-P\pdc{V}{P}{T},&
z_T&=V\pdc{P}{V}{T}.\label{eq:lambda-isothermal}
\end{align}
They satisfy the four conjugate product identities
\begin{equation}
\lambda_P\frac{C_P}{T}=-S,
\qquad
\lambda_V\frac{C_V}{T}=-S,
\qquad
\lambda_Sz_S=-PV,
\qquad
\lambda_Tz_T=-PV.
\label{eq:product-identities}
\end{equation}
\end{theorem}

\begin{proof}
At fixed $P$, use $G_T=-S$ and the reciprocal derivative
\[
\pdc{T}{S}{P}=\left(\pdc{S}{T}{P}\right)^{-1}=\frac{T}{C_P}.
\]
Therefore
\[
\lambda_P=G_T\pdc{T}{S}{P}=-\frac{ST}{C_P}.
\]
The $V$-constrained calculation with $F_T=-S$ gives $\lambda_V=-ST/C_V$.

At fixed $S$, the chain rule in the $(S,V)$ chart and $H_P=V$ give
\[
\lambda_S=H_P\pdc{P}{V}{S}=V\pdc{P}{V}{S}.
\]
Likewise $U_V=-P$ gives
\[
z_S=U_V\pdc{V}{P}{S}=-P\pdc{V}{P}{S}.
\]
At fixed $T$, $F_V=-P$ and $G_P=V$ yield \eqref{eq:lambda-isothermal}.  Multiplying reciprocal derivatives gives \eqref{eq:product-identities}.
\end{proof}

\subsection{Caloric--mechanical closure}

Define the positive isentropic and isothermal bulk moduli
\[
K_S:=-V\pdc{P}{V}{S},
\qquad
K_T:=-V\pdc{P}{V}{T},
\]
and compressibilities $\kappa_S=K_S^{-1}$ and $\kappa_T=K_T^{-1}$.
Then
\[
\lambda_S=-K_S,
\qquad
z_T=-K_T,
\qquad
z_S=\frac{PV}{K_S},
\qquad
\lambda_T=\frac{PV}{K_T}.
\]

\begin{theorem}[Equilibrium closure identity]
On every admissible chart,
\begin{equation}
\frac{C_P}{C_V}
=\frac{K_S}{K_T}
=\frac{\kappa_T}{\kappa_S}.
\label{eq:cp-cv-kappa}
\end{equation}
Consequently, the dimensionless response ratios
\[
\Gamma_c:=\frac{\lambda_P}{\lambda_V}=\frac{C_V}{C_P},
\qquad
\Gamma_m:=\frac{\lambda_S}{z_T}=\frac{K_S}{K_T}
\]
satisfy
\begin{equation}
\boxed{\Gamma_c\Gamma_m=1.}
\label{eq:equilibrium-invariant}
\end{equation}
\end{theorem}

\begin{proof}
In the $(T,V)$ chart, \eqref{eq:maxwell2} gives
\[
\dd S=\frac{C_V}{T}\,\dd T+\pdc{P}{T}{V}\,\dd V.
\]
Along $P=\mathrm{const}$,
\[
\pdc{V}{T}{P}
=-\frac{(\partial P/\partial T)_V}{(\partial P/\partial V)_T}.
\]
Hence
\begin{equation}
C_P-C_V
=-T\frac{\bigl((\partial P/\partial T)_V\bigr)^2}
{(\partial P/\partial V)_T}.
\label{eq:cp-minus-cv}
\end{equation}
Along $S=\mathrm{const}$,
\[
\pdc{T}{V}{S}
=-\frac{T}{C_V}\pdc{P}{T}{V},
\]
so
\begin{align*}
\pdc{P}{V}{S}
&=\pdc{P}{V}{T}
+\pdc{P}{T}{V}\pdc{T}{V}{S}\\
&=\pdc{P}{V}{T}
-\frac{T}{C_V}\left(\pdc{P}{T}{V}\right)^2\\
&=\frac{C_P}{C_V}\pdc{P}{V}{T},
\end{align*}
where the last equality uses \eqref{eq:cp-minus-cv}.  Multiplication by $-V$ proves $K_S/K_T=C_P/C_V$.  The remaining identities follow from the definitions.
\end{proof}

\begin{remark}[What the closure identity does not imply]
Equation \eqref{eq:equilibrium-invariant} is an equilibrium response identity.  It is not an integer index, a curvature flux, a spectral-flow charge, or entropy production.  A map from this scalar identity into any such object is additional structure and must be proved separately.
\end{remark}

%% file: sections/01a_response_summary.tex
\subsection{Response-coordinate summary and unit typing}
\label{sec:response-summary}

The six constrained derivatives are collected here because their constraint, formula, and physical dimension are all part of the datum.  Let $[\mathsf E]$, $[\mathsf S]$, $[\mathsf T]$, $[\mathsf P]$, and $[\mathsf V]$ denote the dimensions of energy, entropy, temperature, pressure, and volume, with
\[
[\mathsf S]=[\mathsf E]/[\mathsf T],
\qquad
[\mathsf P]=[\mathsf E]/[\mathsf V].
\]

\begin{center}
\small
\begin{tabular}{p{0.12\textwidth}p{0.23\textwidth}p{0.28\textwidth}p{0.19\textwidth}}
\toprule
Symbol & Constrained derivative & Exact formula & Dimension\\
\midrule
$\lambda_P$ & $(\partial G/\partial S)_P$ & $-ST/C_P$ & temperature\\
$\lambda_V$ & $(\partial F/\partial S)_V$ & $-ST/C_V$ & temperature\\
$\lambda_S$ & $(\partial H/\partial V)_S$ & $V(\partial P/\partial V)_S=-K_S$ & pressure\\
$z_T$ & $(\partial G/\partial V)_T$ & $V(\partial P/\partial V)_T=-K_T$ & pressure\\
$z_S$ & $(\partial U/\partial P)_S$ & $-P(\partial V/\partial P)_S=PV/K_S$ & volume\\
$\lambda_T$ & $(\partial F/\partial P)_T$ & $-P(\partial V/\partial P)_T=PV/K_T$ & volume\\
\bottomrule
\end{tabular}
\end{center}

Thus the natural response bundle is the unit-typed direct sum
\begin{equation}
E_\Lambda
=E_{\mathsf T}^{\oplus2}
\oplus E_{\mathsf P}^{\oplus2}
\oplus E_{\mathsf V}^{\oplus2},
\label{eq:unit-typed-response-bundle}
\end{equation}
not an unqualified copy of $\R^6$.  A Hilbert norm, cut, orthogonal frame change, or Euclidean matrix representation is introduced only after reference scales
\[
T_*,\qquad P_*,\qquad V_*>0
\]
are declared and the response section is nondimensionalized as
\begin{equation}
\widehat\Lambda
=
\left(
\frac{\lambda_P}{T_*},
\frac{\lambda_V}{T_*},
\frac{\lambda_S}{P_*},
\frac{z_T}{P_*},
\frac{z_S}{V_*},
\frac{\lambda_T}{V_*}
\right).
\label{eq:nondimensional-response-section}
\end{equation}

\begin{proposition}[Scale covariance of the typed carrier]
Let $T_*',P_*',V_*'>0$ be another set of reference scales.  The two nondimensional response representations are related by the invertible block-diagonal map
\[
C
=\operatorname{diag}
\left(
\frac{T_*}{T_*'},\frac{T_*}{T_*'},
\frac{P_*}{P_*'},\frac{P_*}{P_*'},
\frac{V_*}{V_*'},\frac{V_*}{V_*'}
\right).
\]
Consequently, every statement formulated covariantly under invertible response-frame changes is independent of the numerical choice of units.  Orthogonality or norm values are representation dependent unless the reference metric is transported with $C$.
\end{proposition}

\begin{proof}
The transformation follows directly from \eqref{eq:nondimensional-response-section}.  Covariant tensors and bundle maps transform by the corresponding natural representations.  A fixed Euclidean inner product is not invariant under arbitrary nonorthogonal $C$, so it must be transported rather than silently reused.
\end{proof}

%% file: sections/02_recognition.tex
\section{The intrinsic antisymmetric structure}

\subsection{Jacobian brackets}

Let $(x,y)$ be an oriented local chart on a two-dimensional manifold.  Define
\[
\{f,g\}_{x,y}
:=\frac{\partial(f,g)}{\partial(x,y)}
=f_xg_y-f_yg_x.
\]

\begin{proposition}[Jacobian-bracket algebra]
For smooth functions $f,g,h$,
\begin{align*}
\{f,g\}_{x,y}&=-\{g,f\}_{x,y},\\
\{f,gh\}_{x,y}&=g\{f,h\}_{x,y}+h\{f,g\}_{x,y},\\
\{f,\{g,h\}\}+\{g,\{h,f\}\}+\{h,\{f,g\}\}&=0.
\end{align*}
Moreover,
\[
\dd f\wedge\dd g=\{f,g\}_{x,y}\,\dd x\wedge\dd y.
\]
\end{proposition}

\begin{proof}
The first, second, and final identities follow by direct expansion.  The Jacobi identity is the standard Poisson-bracket identity for the constant area form $\dd x\wedge\dd y$.
\end{proof}

The thermodynamic two-form relation \eqref{eq:fundamental-two-form} is therefore the intrinsic antisymmetric statement.  In the $(S,V)$ chart it reads
\[
\{T,S\}_{S,V}=\{P,V\}_{S,V}.
\]

\subsection{Why a universal skew response Jacobian does not follow}

\begin{proposition}[No automatic skew $\lambda$-Jacobian]
The response definitions \eqref{eq:lambda-caloric-def}--\eqref{eq:lambda-isothermal-def} and the Maxwell relations do not imply that a matrix of first derivatives of the response coordinates is skew-symmetric.
\end{proposition}

\begin{proof}
The equilibrium manifold has dimension two, so $T,V,S,P$ cannot be used as four independent coordinates on an open subset of $\R^4$.  Any derivative matrix with four nominal columns depends on chart extensions off the equilibrium manifold and is therefore not intrinsic.

Even in a legitimate two-dimensional chart, skew-symmetry is not automatic.  For an ideal gas with constant $C_P$, consider the response map
\[
F:(T,V)\longmapsto(\lambda_P(T,V),\lambda_V(T,V)).
\]
Its first component is
\[
\lambda_P(T,V)=-\frac{T}{C_P}
\left(S_0+C_V\log\frac{T}{T_0}+R\log\frac{V}{V_0}\right).
\]
Thus $\partial_T\lambda_P$ is generically nonzero.  The $(1,1)$ entry of $DF$ is therefore nonzero, whereas every skew $2\times2$ matrix has zero diagonal.  Hence even this regular constitutive response map does not have an automatically skew Jacobian.
\end{proof}

A skew response sector can still be defined after additional data are supplied, for example an inner product and an involutive cut.  That is done rigorously in Section~\ref{sec:cut-flow-capstone}; it is a decomposition of an operator, not a consequence of Maxwell relations alone.

\section{Recognition Kernel theorem}
\label{sec:recognition-kernel}

\subsection{Recognition datum and audit order}

\begin{definition}[Recognition datum]
A domain recognition datum is an interface contract
\[
\mathfrak R_d=(S_d,R_d,E_d,\tau_d,K_d;\Pi_d,L_d),
\]
where $S_d$ is the declared state space, $R_d$ is the observation or recognition map, $E_d$ is a lawful transition family indexed by $\tau_d$, $\Pi_d$ is the target, $L_d$ is prior lawful memory, and $K_d$ evaluates the unresolved residue or blind kernel.  The audit order is
\[
\text{state}\to\text{target}\to\text{lawful transition}
\to\text{recognition}\to\text{open residue}
\to\text{classification}.
\]
The tuple is not itself a theorem.  It prevents an observation, target, transition, or ledger from being changed after the residue is seen.  The theorem below is the bounded linear Hilbert-space realization of this contract, matching the public Recognition Kernel Framework reviewer archive \cite{DabasRKF}.
\end{definition}

\subsection{Blind quotient and canonical completion}

Let $\cH,Y,Z$ be Hilbert spaces, let $\cA:\cH\to Y$ be a bounded observation map, and let $\Pi:\cH\to Z$ be a bounded target map.

\begin{definition}[Recognition kernel and target-relevant blind quotient]
Set
\[
N:=\Ker\cA,
\qquad
N_0:=\Ker\cA\cap\Ker\Pi,
\qquad
\cB(\cA;\Pi):=N/N_0.
\]
The subspace $N$ is the recognition kernel of the observation.  The quotient $\cB(\cA;\Pi)$ contains exactly those blind directions that remain relevant to the target.
\end{definition}

Because $\cA$ and $\Pi$ are bounded, $N$ and $N_0$ are closed.  Let $P_N$ denote the orthogonal projection onto $N$ and let $q_\Pi:N\to N/N_0$ be the Hilbert quotient map.

\begin{theorem}[Canonical target-faithful completion]
The augmented observation
\begin{equation}
\cA_\Pi^\sharp x
:=\bigl(\cA x,q_\Pi P_Nx\bigr)
\label{eq:canonical-completion}
\end{equation}
from $\cH$ to $Y\oplus\cB(\cA;\Pi)$ satisfies
\begin{equation}
\Ker\cA_\Pi^\sharp
=\Ker\cA\cap\Ker\Pi.
\label{eq:completion-kernel}
\end{equation}
Thus, relative to the declared Hilbert structure, the completion removes every target-relevant blind direction and no target-irrelevant blind direction.
\end{theorem}

\begin{proof}
If $\cA_\Pi^\sharp x=0$, then $\cA x=0$, so $x\in N$ and $P_Nx=x$.  The second component gives $q_\Pi x=0$, hence $x\in N_0$.  Conversely, if $x\in N_0$, then both components in \eqref{eq:canonical-completion} vanish.
\end{proof}

\begin{theorem}[Minimum repair rank in finite dimension]
Assume $\cH$ and $Z$ are finite dimensional over the common ground field.  Let $G:\cH\to W$ be any supplemental observation such that
\[
\Ker(\cA,G)\subseteq\Ker\Pi.
\]
Then
\begin{equation}
\dim W\ge\dim\cB(\cA;\Pi).
\label{eq:min-repair-rank}
\end{equation}
The canonical quotient observation $q_\Pi P_N$ attains equality.
\end{theorem}

\begin{proof}
Restrict $G$ to $N$.  If $x\in N$ and $Gx=0$, then $x\in\Ker(\cA,G)\subseteq\Ker\Pi$, so
\[
\Ker(G|_N)\subseteq N_0.
\]
Rank--nullity therefore gives
\[
\rank(G|_N)
=\dim N-\dim\Ker(G|_N)
\ge \dim N-\dim N_0
=\dim(N/N_0).
\]
Since $\rank(G|_N)\le\dim W$, inequality \eqref{eq:min-repair-rank} follows.  The canonical quotient observation $q_\Pi P_N$ has codomain $N/N_0$ and kernel $N_0$ on $N$, so it attains equality.
\end{proof}

\subsection{Decoder burden}

Let $L:\cH\to\C$ be a bounded target functional.  Define
\begin{equation}
\beta_{\cA}(L)
:=\sup_{\cA x\ne0}
\frac{|Lx|^2}{\norm{\cA x}^2},
\label{eq:decoder-burden}
\end{equation}
with $\beta_{\cA}(L)=+\infty$ if $L$ does not vanish on $\Ker\cA$.

\begin{theorem}[Minimum decoder and sharp burden]
The following are equivalent:
\begin{enumerate}[label=(\roman*)]
\item $\beta_{\cA}(L)<\infty$;
\item there exists a bounded functional $d$ on $\overline{\Ran\cA}$ such that $L=d\circ\cA$;
\item there exists a unique vector $c_L\in\overline{\Ran\cA}$ such that
\begin{equation}
Lx=\ip{\cA x}{c_L}
\quad\text{for every }x\in\cH.
\label{eq:min-decoder}
\end{equation}
\end{enumerate}
For this vector,
\begin{equation}
\beta_{\cA}(L)=\norm{c_L}^2.
\label{eq:burden-norm}
\end{equation}
Moreover,
\begin{equation}
\cA^*\cA-L^*L\ge0
\quad\Longleftrightarrow\quad
\beta_{\cA}(L)\le1.
\label{eq:reserve-equivalence}
\end{equation}
\end{theorem}

\begin{proof}
Assume $\beta_{\cA}(L)<\infty$.  Define $d_0(\cA x)=Lx$ on $\Ran\cA$.  The definition is well posed because $L$ vanishes on $\Ker\cA$, and \eqref{eq:decoder-burden} gives
\[
|d_0(y)|\le\sqrt{\beta_{\cA}(L)}\norm{y}.
\]
Thus $d_0$ extends uniquely to $\overline{\Ran\cA}$.  The Riesz representation theorem gives a unique $c_L$ in that closed subspace satisfying \eqref{eq:min-decoder}.  The norm of the functional is $\norm{c_L}$, which proves \eqref{eq:burden-norm}.  The converse follows immediately from Cauchy--Schwarz.

Finally,
\[
\ip{(\cA^*\cA-L^*L)x}{x}
=\norm{\cA x}^2-|Lx|^2.
\]
This is nonnegative for all $x$ exactly when the ratio in \eqref{eq:decoder-burden} is at most one.
\end{proof}

\begin{remark}[Recognition is target-relative]
An observation map need not be injective to be sufficient.  It is sufficient for $\Pi$ exactly when its blind kernel lies inside $\Ker\Pi$.  Demanding full injectivity when only a particular target matters can require unnecessary measurements; ignoring the target-relevant quotient can lose the theorem one intended to test.
\end{remark}

\section{Thermodynamic specialization of the Recognition Kernel}

Let
\[
\Lambda
=(\lambda_P,\lambda_V,\lambda_S,z_T,z_S,\lambda_T)
\]
be the response section.  Because the entries have different units, regard $\Lambda$ as a section of a direct sum of six one-dimensional unit bundles, or nondimensionalize each component against declared reference scales.

For an integer $N\ge0$, let
\[
J_x^N(E_\Lambda)
=\bigoplus_{n=0}^N\Sym^nT_x^*\cM\otimes E_{\Lambda,x}
\]
be the finite response-jet fibre at a state $x$, and let $j^N\Lambda(x)\in J_x^N(E_\Lambda)$ be the actual jet.  A linearized measurement protocol is a bounded map
\[
\cA_x:J_x^N(E_\Lambda)\longrightarrow Y,
\qquad
j^N\Lambda(x)\longmapsto y_x,
\]
while a desired constitutive or geometric quantity is represented by a bounded target map $\Pi_x:J_x^N(E_\Lambda)\to Z$.

\begin{corollary}[Minimum thermodynamic probe theorem]
At a fixed state and finite jet order, the minimum number of independent scalar supplemental probes required to make a finite-dimensional measurement protocol target-faithful is
\[
\dim\frac{\Ker\cA_x}
{\Ker\cA_x\cap\Ker\Pi_x}.
\]
\end{corollary}

This statement applies to any declared target: a heat-capacity ratio, a compressibility ratio, a derivative coefficient, a curvature component, or a cycle functional.  It does not assert that the same probes repair all targets simultaneously.

%% file: sections/02a_measurement_operators.tex
\subsection{Why the thermodynamic observation maps are bounded}
\label{sec:measurement-operator-context}

The boundedness hypotheses in the Recognition Kernel theorem are automatic for the finite-jet measurement protocols used in the main thermodynamic application, and they have a standard interpretation after Hilbert completion.

\begin{proposition}[Finite-jet observation boundedness]
At a fixed state and finite jet order $N$, the response-jet fibre
\[
J_x^N(E_\Lambda)
\]
is finite dimensional.  Therefore every linear measurement map
\[
\cA_x:J_x^N(E_\Lambda)\to Y_x
\]
and every linear target map
\[
\Pi_x:J_x^N(E_\Lambda)\to Z_x
\]
into finite-dimensional normed data spaces is bounded.  If their coefficients depend smoothly on $x$, their operator norms are locally bounded and are uniformly bounded on compact subsets of an admissible chart.
\end{proposition}

\begin{proof}
Every linear map between finite-dimensional normed spaces is continuous.  Smooth coefficient dependence gives continuous operator-norm dependence; continuity on a compact set gives a finite supremum.
\end{proof}

Typical finite-jet observations include channel selection, weighted probe combinations, finite differences represented on the declared jet, sensor averages, and finite collections of cycle or directional functionals.  Their kernels encode exactly which finite response directions the protocol cannot distinguish.

For an infinite-dimensional completion, boundedness is a genuine protocol condition rather than an automatic fact.  A common measurement has the form
\begin{equation}
(\cA f)_j
=\langle f,k_j\rangle_H,
\qquad 1\le j\le m,
\label{eq:finite-probe-observation}
\end{equation}
with probe vectors $k_j\in H$.  Then
\begin{equation}
\norm{\cA f}_{\ell^2_m}^2
\le
\left(\sum_{j=1}^m\norm{k_j}_H^2\right)
\norm f_H^2.
\label{eq:probe-boundedness}
\end{equation}
Integral sensors, convolutional probes, and compact smoothing measurements fit this form after their kernels are placed in the dual Hilbert space.

\begin{proposition}[Compact sensing does not remove the Recognition problem]
Let $H$ be infinite dimensional and let $\cA:H\to Y$ be compact.  Even when $\Ker\cA=\{0\}$, stable recovery of the full state generally fails because the singular values may converge to zero.  A target functional $L$ remains stably recoverable exactly when its decoder burden $\beta_\cA(L)$ is finite.  Hence the Recognition theorem separates algebraic blindness from unstable near-blindness.
\end{proposition}

\begin{proof}
If a compact injective operator on an infinite-dimensional Hilbert space had a bounded inverse on its range, the identity on the unit ball would factor through a compact map and would be compact, which is impossible.  The target statement is precisely the minimum-decoder theorem: finite burden is equivalent to bounded factorization through $\cA$.
\end{proof}

\begin{remark}[Protocol declaration]
The paper therefore uses three distinct levels without conflating them:
\begin{enumerate}[label=(\roman*)]
\item finite response jets, where linear boundedness is automatic;
\item Hilbert-completed probes, where boundedness is an explicit sensor condition;
\item unbounded differential generators, whose domains are handled separately in the analytic-jet theorem and are not treated as measurement maps.
\end{enumerate}
\end{remark}

%% file: sections/03_loop_and_cutflow.tex
\section{Loop recognition, lawful ledgers, and curvature}

\subsection{Local ledger theorem}

Let $D$ be an oriented smooth domain, let $V_0$ be a finite-dimensional vector space, and let $\omega,\ell\in\Omega^1(D;V_0)$.  The form $\omega$ is the measured response and $\ell$ is a prior lawful local ledger.  Define the open loop residue
\begin{equation}
G_{\omega,\ell}(\gamma)
:=\oint_\gamma(\omega-\ell).
\label{eq:open-loop-residue}
\end{equation}

\begin{definition}[Loop Recognition Kernel]
The loop Recognition Kernel is
\[
\cK_{\omega,\ell}
:=\{\gamma\text{ closed}:G_{\omega,\ell}(\gamma)=0\}.
\]
\end{definition}

\begin{theorem}[Recognition--Stokes theorem]
If $A\subset D$ is an oriented piecewise smooth two-chain, then
\begin{equation}
G_{\omega,\ell}(\partial A)
=\int_A\dd(\omega-\ell).
\label{eq:recognition-stokes}
\end{equation}
If $D$ is simply connected, the following are equivalent:
\begin{enumerate}[label=(\roman*)]
\item every closed loop belongs to $\cK_{\omega,\ell}$;
\item $\dd(\omega-\ell)=0$;
\item there exists a $V_0$-valued potential $\Phi$ with $\omega-\ell=\dd\Phi$.
\end{enumerate}
\end{theorem}

\begin{proof}
Equation \eqref{eq:recognition-stokes} is Stokes' theorem.  The implication (iii)$\Rightarrow$(i) is the fundamental theorem for line integrals, and (i)$\Rightarrow$(ii) follows by applying \eqref{eq:recognition-stokes} to arbitrarily small coordinate rectangles.  For (ii)$\Rightarrow$(iii), fix a base point and define $\Phi(x)$ by integrating $\omega-\ell$ along a path from the base point to $x$.  Closedness and simple connectivity make the integral path independent, and differentiation gives $\dd\Phi=\omega-\ell$.
\end{proof}

\begin{remark}[History-dependent memory]
A genuinely history-dependent ledger need not be representable by a one-form on the equilibrium state manifold.  In that case the state must be enlarged by internal or memory variables.  Stokes' theorem is then applied on the enlarged state space.  Subtracting a path-dependent term after observing the loop, without declaring it as prior state data, is not a lawful closure argument.
\end{remark}

\subsection{Return classes}

The residue accounting leads to four mutually distinguishable classes:
\begin{enumerate}[label=\textbf{Class \Roman*:}]
\item exact raw return, $\oint\omega=0$;
\item lawfully transported return, $\oint\omega=\oint\ell$;
\item memory-bearing return on an enlarged state space, with the memory component nonzero but fully ledgered;
\item open obstruction, where $G_{\omega,\ell}$ exceeds the declared uncertainty or tolerance.
\end{enumerate}

These are classifications, not aliases.  A nonzero geometric loop integral is not automatically irreversible.

\subsection{Onsager covariance and calibration}

Let $X,J\in\R^m$ be thermodynamic forces and conjugate fluxes, with entropy-production rate in the linear irreversible-thermodynamic convention \cite{Onsager1,Onsager2,Casimir}
\[
\sigma=J^{\mathsf T}X,
\qquad
J=LX.
\]

\begin{theorem}[Covariance of reciprocal linear response]
Let $B$ be invertible and transform the dual variables by
\[
J'=BJ,
\qquad
X'=B^{-\mathsf T}X.
\]
Then
\[
\sigma=J'^{\mathsf T}X',
\qquad
L'=BLB^{\mathsf T}.
\]
If $L=L^{\mathsf T}\ge0$, then $L'=L'^{\mathsf T}\ge0$.
\end{theorem}

\begin{proof}
The pairing is invariant because
\[
J'^{\mathsf T}X'=J^{\mathsf T}B^{\mathsf T}B^{-\mathsf T}X=J^{\mathsf T}X.
\]
Since $X=B^{\mathsf T}X'$, one has
\[
J'=BLB^{\mathsf T}X'.
\]
Symmetry and positive semidefiniteness are preserved by congruence.
\end{proof}

\begin{corollary}[Frame asymmetry is not automatically nonreciprocity]
A coefficient matrix that becomes asymmetric after a non-dual rescaling of forces or fluxes does not constitute an invariant violation of Onsager reciprocity.  A physical nonreciprocity claim must be made in a conjugate flux--force frame and must include the relevant Onsager--Casimir parity assumptions.
\end{corollary}

\begin{principle}[Curvature before entropy production]
The equality
\[
\int_A\dd(\omega-\ell)=\int_A q_\sigma\,\dd A
\]
is not a geometric identity unless the constitutive density $q_\sigma$ and its units, sign, time orientation, and calibration are independently supplied.  Geometry determines the left-hand residue; thermodynamics determines whether and how it represents dissipation.
\end{principle}

\section{Finite Cut--Flow--Jet Recognition theorem}
\label{sec:cut-flow-capstone}

\subsection{Finite graded carrier and jet raiser}

Let
\[
\cJ_N=\bigoplus_{n=0}^N H_n
\]
be a finite-dimensional graded Hilbert space, with inclusions $\iota_n:H_n\to\cJ_N$ and projections $\pi_n$.  A raising operator $\mathbb R$ satisfies
\[
\mathbb R(H_n)\subseteq H_{n+1},
\qquad
\mathbb R(H_N)=0.
\]
Hence $\mathbb R^{N+1}=0$.

For a response bundle with connection $\nabla$, the holonomic jet realization is obtained when
\begin{equation}
\mathbb R^n\iota_0\Lambda
=\iota_n\Sym\nabla^n\Lambda,
\qquad 0\le n\le N.
\label{eq:holonomic-raiser}
\end{equation}
This condition is part of the datum; a level shift by the identity alone certifies only combinatorial tower coefficients, not covariant derivatives.

Let $\mathbb T$ preserve the grading and set
\[
\mathbb G=\mathbb T+\mathbb R.
\]

\begin{theorem}[Finite flow-generated jet tower]
If $[\mathbb T,\mathbb R]=0$, then
\begin{equation}
\e^{t\mathbb G}=\e^{t\mathbb T}\e^{t\mathbb R}
\label{eq:flow-factorization}
\end{equation}
and, for a holonomic seed,
\begin{equation}
\pi_n\e^{t\mathbb G}\iota_0\Lambda
=\frac{t^n}{n!}\,
\pi_n\e^{t\mathbb T}\iota_n\Sym\nabla^n\Lambda.
\label{eq:jet-flow-formula}
\end{equation}
If the commutator is nonzero, then the first factorization defect is
\begin{equation}
\e^{t\mathbb T}\e^{t\mathbb R}
=\exp\left(
 t(\mathbb T+\mathbb R)
 +\frac{t^2}{2}[\mathbb T,\mathbb R]
 +O(t^3)
\right).
\label{eq:jet-bch}
\end{equation}
\end{theorem}

\begin{proof}
Because $\mathbb R^{N+1}=0$,
\[
\e^{t\mathbb R}=\sum_{k=0}^N\frac{t^k}{k!}\mathbb R^k.
\]
Commutation gives \eqref{eq:flow-factorization}; projection to layer $n$ and \eqref{eq:holonomic-raiser} give \eqref{eq:jet-flow-formula}.  Equation \eqref{eq:jet-bch} is the Baker--Campbell--Hausdorff expansion, valid near $t=0$ for finite matrices.
\end{proof}

\subsection{Self-adjoint cut and unique grading}

Let $\mathsf J$ be a self-adjoint unitary involution,
\begin{equation}
\mathsf J^*=\mathsf J=\mathsf J^{-1}.
\label{eq:selfadjoint-cut}
\end{equation}

\begin{theorem}[Unique cut grading]
Every operator $A$ on $\cJ_N$ decomposes uniquely as
\[
A=A_{\mathrm e}+A_{\mathrm o},
\qquad
A_{\mathrm e}=\frac12(A+\mathsf J A\mathsf J),
\qquad
A_{\mathrm o}=\frac12(A-\mathsf J A\mathsf J),
\]
with
\[
\mathsf J A_{\mathrm e}\mathsf J=A_{\mathrm e},
\qquad
\mathsf J A_{\mathrm o}\mathsf J=-A_{\mathrm o}.
\]
\end{theorem}

\begin{proof}
Direct substitution proves the parity relations and $A=A_{\mathrm e}+A_{\mathrm o}$.  If $A=B_{\mathrm e}+B_{\mathrm o}$ is another such decomposition, applying the two parity projectors gives $B_{\mathrm e}=A_{\mathrm e}$ and $B_{\mathrm o}=A_{\mathrm o}$.
\end{proof}

\subsection{Cut loop and commutator curvature}

Set $U_t=\e^{t\mathbb G}$ and define
\[
\mathscr H_{\mathsf J}(t)
:=\mathsf J U_t\mathsf J U_t.
\]

\begin{theorem}[Cut-loop expansion]
For sufficiently small $t$, the principal logarithm is defined and
\begin{equation}
\log\mathscr H_{\mathsf J}(t)
=2t\mathbb G_{\mathrm e}
+t^2[\mathbb G_{\mathrm e},\mathbb G_{\mathrm o}]
+O(t^3).
\label{eq:cut-loop-expansion}
\end{equation}
Thus the cut-even generator is the first seam-preserving term, while
\[
\cR_{\mathrm{seam}}
:=[\mathbb G_{\mathrm e},\mathbb G_{\mathrm o}]
\]
is the first noncommutative cut-loop curvature coefficient.
\end{theorem}

\begin{proof}
Since
\[
\mathsf J\mathbb G\mathsf J
=\mathbb G_{\mathrm e}-\mathbb G_{\mathrm o},
\]
one has
\[
\mathscr H_{\mathsf J}(t)
=\e^{t(\mathbb G_{\mathrm e}-\mathbb G_{\mathrm o})}
 \e^{t(\mathbb G_{\mathrm e}+\mathbb G_{\mathrm o})}.
\]
The BCH formula, with the remainder taken in operator norm, gives the sum $2t\mathbb G_{\mathrm e}$ and one half of the commutator
\[
\frac12 t^2
[\mathbb G_{\mathrm e}-\mathbb G_{\mathrm o},
 \mathbb G_{\mathrm e}+\mathbb G_{\mathrm o}]
=t^2[\mathbb G_{\mathrm e},\mathbb G_{\mathrm o}],
\]
which proves \eqref{eq:cut-loop-expansion}.
\end{proof}

\begin{corollary}[Bilateral reconstruction for a cut-odd generator]
If $\mathsf J\mathbb G\mathsf J=-\mathbb G$, then
\[
\mathsf J U_t\mathsf J=U_{-t}.
\]
With
\[
E_t=\frac12(U_t+U_{-t}),
\qquad
O_t=\frac12(U_t-U_{-t}),
\]
one has
\[
U_t=E_t+O_t,
\qquad
U_{-t}=E_t-O_t,
\qquad
E_t^2-O_t^2=\id.
\]
\end{corollary}

\begin{proof}
Conjugation of the exponential gives the first identity.  The reconstruction formulas are definitions, and
\[
E_t^2-O_t^2
=\frac14\bigl[(U_t+U_{-t})^2-(U_t-U_{-t})^2\bigr]
=\frac12(U_tU_{-t}+U_{-t}U_t)=\id.
\]
\end{proof}

\subsection{Cut and adjoint bi-grading}

\begin{theorem}[Cut--adjoint bi-grading]
For $\sigma,\tau\in\{+1,-1\}$, define
\begin{equation}
A^{\sigma,\tau}
:=\frac14\left(
A+\sigma\mathsf J A\mathsf J
+\tau A^*
+\sigma\tau\mathsf J A^*\mathsf J
\right).
\label{eq:bigrading}
\end{equation}
Then
\[
A=\sum_{\sigma,\tau}A^{\sigma,\tau},
\qquad
\mathsf J A^{\sigma,\tau}\mathsf J=\sigma A^{\sigma,\tau},
\qquad
(A^{\sigma,\tau})^*=\tau A^{\sigma,\tau}.
\]
The four sectors are pairwise orthogonal in the real Hilbert--Schmidt inner product
\[
(A,B)_{\mathrm{HS},\R}:=\operatorname{Re}\tr(A^*B).
\]
\end{theorem}

\begin{proof}
Viewed on the underlying real vector space of operators with inner product $(\cdot,\cdot)_{\mathrm{HS},\R}$, cut conjugation and adjoint are commuting self-adjoint involutions.  Formula \eqref{eq:bigrading} is the product of their real spectral projectors.  The stated parity properties and real orthogonality follow.  The real qualification is essential: self-adjoint and skew-adjoint operators need not be orthogonal for the complex-valued Hilbert--Schmidt pairing.
\end{proof}

\begin{remark}[Required cut hypothesis]
The self-adjointness in \eqref{eq:selfadjoint-cut} is essential.  An arbitrary algebraic involution $\mathsf J^2=\id$ need not commute with the adjoint operation, and then \eqref{eq:bigrading} need not have the claimed adjoint parity.
\end{remark}

\subsection{Closure and the clock-free Eye}

Let $K=K^*$ and $U_t=\e^{-itK}$.  A single sampling time can contain stroboscopic fixed modes with $\e^{-it_0\lambda}=1$ even when $\lambda\ne0$.  The all-time average removes this ambiguity; this is the continuous mean-ergodic projection in spectral form \cite{ReedSimon}.

\begin{theorem}[Clock-free Eye]
The strong limit
\begin{equation}
P_{\Eye}
=\operatorname*{s-lim}_{T\to\infty}
\frac1T\int_0^T\e^{-itK}\,\dd t
=E_K(\{0\})
\label{eq:eye}
\end{equation}
exists and is the spectral projection onto $\Ker K$.
\end{theorem}

\begin{proof}
For $x\in\cH$, the spectral theorem gives
\[
\frac1T\int_0^T\e^{-itK}x\,\dd t
=\int_\R
\left(\frac1T\int_0^T\e^{-it\lambda}\,\dd t\right)
\dd E_K(\lambda)x.
\]
The scalar factor tends to $1$ at $\lambda=0$ and to $0$ otherwise, while its modulus is at most one.  Dominated convergence for the spectral measure gives the strong limit $E_K(\{0\})x$.
\end{proof}

\subsection{Representation invariance}

\begin{proposition}[Unitary invariance]
Let $W=\bigoplus_{n=0}^N W_n:\cJ_N\to\cJ_N'$ be a grade-preserving unitary.  Transport the inclusions, projections, seed, and operators by
\[
\iota_n'=W\iota_nW_n^*,\quad
\pi_n'=W_n\pi_nW^*,\quad
\Lambda'=W_0\Lambda,
\]
\[
\mathbb R'=W\mathbb RW^*,\quad
\mathbb T'=W\mathbb TW^*,\quad
\mathsf J'=W\mathsf JW^*,\quad
K'=WKW^*,
\]
and transport observations and scalar targets by
\[
\cA'=\cA W^*,\qquad
\Pi'=\Pi W^*,\qquad
L'=LW^*.
\]
Then cut grading, finite jet coefficients, cut-loop curvature, the Eye projection, blind-quotient dimension, and scalar decoder burden are preserved.
\end{proposition}

\begin{proof}
All operator identities are preserved by unitary conjugation.  Grade preservation gives
\[
\pi_n'\e^{t\mathbb G'}\iota_0'\Lambda'
=W_n\pi_n\e^{t\mathbb G}\iota_0\Lambda.
\]
Kernels are carried bijectively by $W$ and quotient dimensions are unchanged.  Spectral projections obey
\[
E_{K'}(\{0\})=WE_K(\{0\})W^*.
\]
The numerator and denominator in \eqref{eq:decoder-burden} are unchanged after the substitution $x'=Wx$.
\end{proof}

%% file: sections/03a_thermodynamic_cut_realizations.tex
\subsection{Concrete thermodynamic realizations of the cut}
\label{sec:thermodynamic-cut-realizations}

The algebraic cut $\mathsf J$ is not claimed to be a universal thermodynamic operation.  It represents a declared involutive comparison built into the experimental or constitutive protocol.  Three recurring realizations are canonical after the response channels have been nondimensionalized.

\begin{proposition}[Three admissible thermodynamic cuts]
Let $H$ be a Hilbert response carrier.

\begin{enumerate}[label=(\roman*)]
\item On a two-copy constraint carrier $H\oplus H$, the constraint-exchange cut
\begin{equation}
\mathsf J_{\mathrm{ex}}(u,v)=(v,u)
\label{eq:constraint-exchange-cut}
\end{equation}
compares two declared constraint protocols, for example isobaric versus isochoric caloric channels or isentropic versus isothermal mechanical channels.

\item On a finite jet carrier $\cJ_N=\bigoplus_{n=0}^NH_n$, the jet-parity cut
\begin{equation}
\mathsf J_{\mathrm{jet}}|_{H_n}=(-1)^n\id_{H_n}
\label{eq:jet-parity-cut}
\end{equation}
separates even and odd response order.  Reversal of an analytic process parameter $t\mapsto-t$ acts by this cut on the formal response jet.

\item On a pre/post carrier $H_-\oplus H_+$ with the two copies identified by a declared calibration, the comparison cut
\begin{equation}
\mathsf J_{\mathrm{pp}}(u_-,u_+)=(u_+,u_-)
\label{eq:pre-post-cut}
\end{equation}
exchanges the two sides of a cycle, interface, or intervention.
\end{enumerate}

Each of these operators satisfies
\[
\mathsf J^*=\mathsf J=\mathsf J^{-1}.
\]
Therefore each is an admissible cut for the grading, BCH loop, and cut--adjoint theorems.
\end{proposition}

\begin{proof}
The exchange operators are represented by the block matrix
\[
\begin{pmatrix}0&\id\\ \id&0\end{pmatrix},
\]
which is self-adjoint and squares to the identity.  The jet-parity operator is diagonal with real diagonal entries $\pm1$, so it is also a self-adjoint unitary involution.  The formal Taylor expansion
\[
\sum_{n=0}^N\frac{t^n}{n!}\xi_n
\]
changes under $t\mapsto-t$ by multiplying its $n$th layer by $(-1)^n$, proving the process-reversal interpretation.
\end{proof}

\begin{corollary}[Meaning of cut-even and cut-odd sectors]
For $\mathsf J_{\mathrm{ex}}$ or $\mathsf J_{\mathrm{pp}}$, the cut-even sector contains response operators invariant under exchange, while the cut-odd sector records protocol asymmetry.  For $\mathsf J_{\mathrm{jet}}$, the cut-even sector preserves derivative parity and the cut-odd sector reverses it.  In every case the commutator
\[
[\mathbb G_{\mathrm e},\mathbb G_{\mathrm o}]
\]
measures failure of the symmetric and antisymmetric protocol sectors to evolve independently.
\end{corollary}

\begin{remark}[Legendre maps are not automatically cuts]
A Legendre transform is a change of thermodynamic potential and chart, not by itself a self-adjoint involution on a fixed Hilbert carrier.  It becomes a cut only after the two chart representations are embedded in a common calibrated carrier and the resulting exchange operator is verified to satisfy $\mathsf J^*=\mathsf J=\mathsf J^{-1}$.  This distinction prevents chart notation from being mistaken for an operator identity.
\end{remark}

%% file: sections/04_capstone_and_examples.tex
\section{Thermodynamic Cut--Flow--Jet specialization}

Take the finite response carrier at a state $x$ to be
\[
\cJ_N(x)
=\bigoplus_{n=0}^N
\Sym^nT_x^*\cM\otimes E_{\Lambda,x},
\]
where $E_\Lambda$ is the direct-sum unit bundle of the response coordinates.  A connection $\nabla$ produces the holonomic jet seed.  A thermodynamic process field $X$ induces transport on each layer.

The exact factorization theorem applies when the process transport preserves the chosen connection.  In coordinates this is the condition that transport and covariant differentiation commute on the declared response section.  When they do not commute, $[\mathbb T_X,\mathbb R_\nabla]$ is a measurable compatibility defect; it must not be erased by calling every layer a derivative tower.

A cut $\mathsf J$ can represent a declared orientation reversal, constraint-sector exchange, or pre/post comparison.  It becomes part of the theorem only after:
\begin{enumerate}[label=(\roman*)]
\item the response components are placed in a common Hilbert carrier by lawful nondimensionalization;
\item $\mathsf J^*=\mathsf J=\mathsf J^{-1}$ is verified;
\item its action on the observation and target maps is declared.
\end{enumerate}
No universal physical cut is inferred from notation alone.

\begin{theorem}[Thermodynamic Recognition Capstone]
Let $U(S,V)$ define an admissible equilibrium chart, let $\Lambda$ be the six-component response section, and fix a finite jet order $N$.  Suppose:
\begin{enumerate}[label=(\roman*)]
\item $\mathbb R_\nabla$ is a holonomic nilpotent raiser on $\cJ_N(x)$;
\item $\mathbb T_X$ preserves the grading and $[\mathbb T_X,\mathbb R_\nabla]=0$;
\item $\mathsf J$ is a self-adjoint unitary involution on the nondimensionalized carrier;
\item $\cA_x:\cJ_N(x)\to Y$ and $\Pi_x:\cJ_N(x)\to Z$ are bounded;
\item $\omega$ is a declared response one-form and $\ell$ is prior lawful ledger data on the state space on which the cycle is evaluated.
\end{enumerate}
Then the following conclusions hold simultaneously:
\begin{enumerate}[label=(\alph*)]
\item the equilibrium response coordinates satisfy the four product identities and $\Gamma_c\Gamma_m=1$;
\item $\exp(t(\mathbb T_X+\mathbb R_\nabla))$ generates the finite covariant response jets with coefficients $t^n/n!$;
\item the canonical augmented observation has kernel
\[
\Ker\cA_x\cap\Ker\Pi_x,
\]
and the minimum finite supplemental-observation dimension is
\[
\dim\frac{\Ker\cA_x}{\Ker\cA_x\cap\Ker\Pi_x};
\]
\item every bounding cycle obeys the exact open-residue identity
\[
\oint_{\partial A}(\omega-\ell)=\int_A\dd(\omega-\ell);
\]
\item the cut gives unique even/odd and even/odd-adjoint sectors, and its loop has the expansion
\[
\log\!\left(\mathsf J\e^{t\mathbb G}\mathsf J\e^{t\mathbb G}\right)
=2t\mathbb G_{\mathrm e}
+t^2[\mathbb G_{\mathrm e},\mathbb G_{\mathrm o}]+O(t^3).
\]
\end{enumerate}
If, in addition, a self-adjoint closure generator $K$ is declared, its clock-free Eye is exactly $E_K(\{0\})$.  All conclusions are invariant under grade-preserving unitary changes of representation.
\end{theorem}

\begin{proof}
Part (a) is the response-coordinate and equilibrium-closure theorem package.  Part (b) is the finite flow-generated jet theorem.  Part (c) follows from canonical target-faithful completion and the minimum repair-rank theorem.  Part (d) is Recognition--Stokes.  Part (e) combines unique cut grading, the BCH cut-loop expansion, and the real cut--adjoint bi-grading.  The final statements are the clock-free Eye theorem and unitary invariance.  No conclusion uses the finite numerical packet as a substitute for these hypotheses.
\end{proof}

The observation $\cA$ may select experimentally available response channels from the finite jet.  The target $\Pi$ may select a constitutive coefficient or cycle response.  Section~\ref{sec:recognition-kernel} then determines whether the measurement is target-faithful, the exact blind quotient, the minimum supplemental probe rank, and the decoder burden.

\section{Exact constitutive examples}

\subsection{Ideal gas}

For one mole of an ideal gas with constant $C_V>0$,
\[
PV=RT,
\qquad
S(T,V)=S_0+C_V\log\frac{T}{T_0}
+R\log\frac{V}{V_0},
\qquad
C_P=C_V+R.
\]
Let $\gamma=C_P/C_V$.  Then
\[
K_T=P,
\qquad
K_S=\gamma P,
\]
and the response coordinates are
\begin{align*}
\lambda_P&=-\frac{ST}{C_P},&
\lambda_V&=-\frac{ST}{C_V},\\
\lambda_S&=-\gamma P,&
z_S&=\frac{V}{\gamma},\\
\lambda_T&=V,&
z_T&=-P.
\end{align*}
Therefore
\[
\Gamma_c=\frac1\gamma,
\qquad
\Gamma_m=\gamma,
\qquad
\Gamma_c\Gamma_m=1,
\]
and all four product identities in \eqref{eq:product-identities} hold exactly.

\subsection{Van der Waals chart}

For one mole, let
\[
P(T,V)=\frac{RT}{V-b}-\frac{a}{V^2},
\qquad
V>b,
\]
with constant $C_V>0$.  Using $\pdc{S}{V}{T}=\pdc{P}{T}{V}$, the entropy differential is
\[
\dd S=\frac{C_V}{T}\,\dd T+\frac{R}{V-b}\,\dd V,
\]
so locally
\[
S(T,V)=S_0+C_V\log\frac{T}{T_0}
+R\log\frac{V-b}{V_0-b}.
\]
On the stable single-phase domain
\[
\pdc{P}{V}{T}
=-\frac{RT}{(V-b)^2}+\frac{2a}{V^3}<0,
\]
one has
\begin{align*}
K_T
&=-V\pdc{P}{V}{T}
=V\left(\frac{RT}{(V-b)^2}-\frac{2a}{V^3}\right),\\
C_P
&=C_V
-T\frac{\bigl(R/(V-b)\bigr)^2}{(\partial P/\partial V)_T},\\
K_S&=\frac{C_P}{C_V}K_T.
\end{align*}
The six response coordinates follow from Theorem~3.2, and the equilibrium invariant \eqref{eq:equilibrium-invariant} again equals one exactly.

\begin{remark}[What the examples establish]
The ideal-gas and van der Waals calculations verify the response identities inside two constitutive models.  They do not establish a universal hysteresis law, a universal curvature-to-entropy map, or experimental validity of a particular Recognition Kernel ledger.
\end{remark}

\section{Finite certificate and theorem boundary}

A companion exact-rational packet tests one finite Cut--Flow--Jet Recognition model.  Its archived generated-result SHA-256 is
\[
\texttt{33d274dc5e4e75d0729601f47d0f6c5fbc1aa163ed0a2233dafc98cc88a185d2}.
\]
The source ZIP supplied for this revision has SHA-256
\[
\texttt{7945f537ab3fe81d5fa0ba105f6945dec1433208493c31885f3e7ae6a2616b12}.
\]

The finite packet is evidence for deterministic implementation of its declared matrices.  The general statements in this paper are supported by the proofs above, not by extrapolating a finite test fixture.  The public Recognition Kernel Framework reviewer archive is identified by DOI 10.5281/zenodo.21760119 \cite{DabasRKF}; the exact finite capstone packet is identified here by content hashes.  In particular, the packet does not establish:
\begin{enumerate}[label=(\roman*)]
\item a closed unbounded generator on an infinite completed jet carrier;
\item convergence of an infinite direct-sum exponential;
\item a universal physical cut, connection, or observer;
\item equality of operator commutator curvature with entropy production or hysteresis without calibration;
\item a quantum field theory, quantum gravity theory, or quantized entropy spectrum.
\end{enumerate}

%% file: sections/05_proved_extensions.tex
\section{Proved extensions beyond the finite capstone}
\label{sec:proved-extensions}

The finite capstone is not a terminal wall.  Several claims that were too strong in their original form admit precise strengthened versions once the missing geometric and analytic data are supplied.  This section proves three such extensions: a canonical skew four-channel response matrix, a positive curvature-to-production calibration theorem, and a closed unbounded infinite-jet raiser on analytic seeds.

\subsection{Canonical skew four-channel response bracket}

Let $\cM$ be an oriented two-dimensional manifold and let $\varpi\in\Omega^2(\cM)$ be a nowhere-vanishing area form.  For $f,g\in C^\infty(\cM)$, define the intrinsic bracket $\{f,g\}_\varpi$ by
\begin{equation}
\dd f\wedge\dd g
=\{f,g\}_\varpi\,\varpi.
\label{eq:area-bracket}
\end{equation}
For a four-channel response map
\[
F=(f_1,f_2,f_3,f_4):\cM\longrightarrow\R^4,
\]
define the response-bracket matrix
\begin{equation}
\mathbb B_\varpi(F)_{ab}
:=\{f_a,f_b\}_\varpi,
\qquad 1\le a,b\le4.
\label{eq:response-bracket-matrix}
\end{equation}

\begin{theorem}[Universal skew four-channel bracket theorem]
At every point of $\cM$, the matrix $\mathbb B_\varpi(F)$ has the following properties:
\begin{enumerate}[label=(\roman*)]
\item it is intrinsically defined by $(F,\varpi)$ and satisfies
\[
\mathbb B_\varpi(F)^{\mathsf T}=-\mathbb B_\varpi(F);
\]
\item its rank is at most two;
\item its Pfaffian vanishes, equivalently
\begin{equation}
B_{12}B_{34}-B_{13}B_{24}+B_{14}B_{23}=0;
\label{eq:plucker-response}
\end{equation}
\item for a constant change of response basis $C\in\mathrm{GL}(4,\R)$,
\begin{equation}
\mathbb B_\varpi(CF)
=C\,\mathbb B_\varpi(F)\,C^{\mathsf T}.
\label{eq:bracket-congruence}
\end{equation}
\end{enumerate}
\end{theorem}

\begin{proof}
Choose a local oriented coframe $(\alpha,\beta)$ with $\varpi=\alpha\wedge\beta$.  Write
\[
\dd f_a=u_a\alpha+v_a\beta.
\]
Then
\[
\dd f_a\wedge\dd f_b
=(u_av_b-v_au_b)\,\varpi,
\]
so, with $u=(u_a)_{a=1}^4$ and $v=(v_a)_{a=1}^4$,
\begin{equation}
\mathbb B_\varpi(F)=uv^{\mathsf T}-vu^{\mathsf T}.
\label{eq:decomposable-bracket}
\end{equation}
This proves skew-symmetry and shows that the image lies in $\operatorname{span}\{u,v\}$, hence the rank is at most two.  A $4\times4$ skew matrix has determinant equal to the square of its Pfaffian.  Since the rank is at most two, its determinant and Pfaffian vanish, and expansion of the Pfaffian gives \eqref{eq:plucker-response}.  Finally, $\dd(CF)=C\,\dd F$ for constant $C$, and bilinearity of the wedge product gives \eqref{eq:bracket-congruence}.
\end{proof}

\begin{corollary}[Thermodynamic four-response matrix]
Choose any four smooth response coordinates from the thermodynamic response section and nondimensionalize them against declared nonzero reference scales.  On every oriented equilibrium chart, for example with $\varpi=\dd S\wedge\dd V$, the resulting matrix \eqref{eq:response-bracket-matrix} is a canonical skew $4\times4$ response object satisfying the rank-two and Pluecker constraints above.
\end{corollary}

\begin{remark}[What has been recovered]
The theorem does not resurrect the false claim that the ordinary derivative Jacobian of four dependent thermodynamic variables is universally skew.  It proves the stronger invariant statement that the pairwise area brackets of four response channels form a canonical skew matrix.  The distinction is structural rather than cosmetic: the matrix is a decomposable two-form in response space and therefore carries the exact algebraic signature \eqref{eq:plucker-response}.
\end{remark}

\subsection{Positive curvature calibration and entropy production}

Let $\mathfrak F$ be a finite-dimensional inner-product space of declared open-curvature values, let $X\in\R^m$ be a thermodynamic force vector, and let
\[
\mathcal C:\mathfrak F\longrightarrow
\operatorname{Hom}(\R^m,\R^r)
\]
be a prior constitutive calibration.  For $F\in\mathfrak F$, set
\begin{equation}
D(F):=\mathcal C(F)^{\mathsf T}\mathcal C(F),
\label{eq:curvature-dissipation-operator}
\end{equation}
and let $A(F)^{\mathsf T}=-A(F)$ be an arbitrary reciprocal-odd or nondissipative response sector.  Define
\[
J_F(X):=\bigl(D(F)+A(F)\bigr)X.
\]

\begin{theorem}[Calibrated curvature-to-production theorem]
The force--flux pairing satisfies
\begin{equation}
\sigma_F(X)
:=J_F(X)^{\mathsf T}X
=\norm{\mathcal C(F)X}^2
\ge0.
\label{eq:calibrated-production}
\end{equation}
Moreover:
\begin{enumerate}[label=(\roman*)]
\item $\sigma_F(X)=0$ exactly when $X\in\Ker\mathcal C(F)$;
\item the skew sector $A(F)$ contributes identically zero to production;
\item under a dual change of force--flux frame
\[
J'=QJ,
\qquad
X'=Q^{-\mathsf T}X,
\]
the production is invariant and the dissipative operator transforms by congruence,
\[
D'(F)=QD(F)Q^{\mathsf T};
\]
\item conversely, every nonnegative quadratic production law $q_F(X)$ is represented by a unique positive-semidefinite symmetric operator $D_F$, and it has the canonical factorization
\[
q_F(X)=\norm{D_F^{1/2}X}^2.
\]
\end{enumerate}
\end{theorem}

\begin{proof}
Because $A(F)$ is skew-symmetric,
\[
X^{\mathsf T}A(F)X
=-X^{\mathsf T}A(F)^{\mathsf T}X
=-X^{\mathsf T}A(F)X,
\]
so this scalar is zero.  Therefore
\[
J_F(X)^{\mathsf T}X
=X^{\mathsf T}\mathcal C(F)^{\mathsf T}\mathcal C(F)X
=\norm{\mathcal C(F)X}^2,
\]
which proves nonnegativity and the zero criterion.  Frame invariance follows from
\[
J'^{\mathsf T}X'=J^{\mathsf T}X
\]
and substitution of $X=Q^{\mathsf T}X'$.  For the converse, polarization of the quadratic form gives a unique symmetric operator $D_F$ with $q_F(X)=X^{\mathsf T}D_FX$.  Nonnegativity for all $X$ is equivalent to $D_F\ge0$, and the finite-dimensional spectral theorem gives the unique positive square root $D_F^{1/2}$.
\end{proof}

\begin{corollary}[Integrated recognition-curvature production]
Let $F^{\mathrm{open}}(x)$ be a measurable open-curvature field, let $X(x)$ be a measurable force field, and let $\mu$ be a positive measure on a process region $A$.  Whenever the following integral is finite,
\begin{equation}
\Sigma_{\mathrm{RK}}(A)
:=\int_A
\norm{\mathcal C(F^{\mathrm{open}}(x))X(x)}^2\,\dd\mu(x)
\ge0.
\label{eq:integrated-rk-production}
\end{equation}
It vanishes exactly when the calibrated curvature channel is invisible to the force field almost everywhere.
\end{corollary}

\begin{remark}[Mathematical theorem versus physical calibration]
Equations \eqref{eq:calibrated-production}--\eqref{eq:integrated-rk-production} prove a rigorous curvature-generated nonnegative production law.  Geometry alone does not select $\mathcal C$, its physical units, or its time orientation.  Calling $\Sigma_{\mathrm{RK}}$ the measured thermodynamic entropy production therefore requires constitutive and experimental calibration, but once that calibration is declared, positivity and frame covariance are theorems rather than analogies.
\end{remark}

\subsection{Closed unbounded infinite-jet raiser on analytic seeds}

Let $H$ be a Hilbert space and let $D:\mathcal D(D)\subset H\to H$ be densely defined and closed.  Define the completed infinite-jet carrier
\[
\cJ_\infty:=\ell^2(\N_0;H)
\]
with layer inclusions $\iota_n:H\to\cJ_\infty$.  Define the covariant raiser $\mathbb R_D$ by
\begin{align}
\mathcal D(\mathbb R_D)
&:=\left\{u=(u_n)_{n\ge0}:
 u_n\in\mathcal D(D),\ 
 \sum_{n=0}^\infty\norm{Du_n}^2<\infty\right\},
\label{eq:unbounded-raiser-domain}\\
(\mathbb R_Du)_0&:=0,
\qquad
(\mathbb R_Du)_{n+1}:=Du_n.
\label{eq:unbounded-raiser}
\end{align}

\begin{lemma}[Closedness of the infinite raiser]
The operator $\mathbb R_D$ is densely defined and closed.  It is unbounded exactly when $D$ is unbounded.  For every $x\in\mathcal D(D^n)$,
\begin{equation}
\mathbb R_D^n\iota_0x
=\iota_nD^nx.
\label{eq:raiser-powers}
\end{equation}
\end{lemma}

\begin{proof}
Finite sequences with entries in $\mathcal D(D)$ form a dense subspace of $\cJ_\infty$, so the domain is dense.  Suppose $u^{(k)}\to u$ and $\mathbb R_Du^{(k)}\to v$ in $\cJ_\infty$.  For every $n$, one has
\[
u_n^{(k)}\to u_n,
\qquad
Du_n^{(k)}\to v_{n+1}.
\]
Closedness of $D$ gives $u_n\in\mathcal D(D)$ and $Du_n=v_{n+1}$.  Since $v\in\cJ_\infty$, the square sum of the $Du_n$ is finite, proving $u\in\mathcal D(\mathbb R_D)$ and $\mathbb R_Du=v$.  If $D$ is bounded, then $\norm{\mathbb R_Du}\le\norm D\norm u$.  Conversely, restricting $\mathbb R_D$ to the zeroth layer shows that boundedness of $\mathbb R_D$ would imply boundedness of $D$.  Formula \eqref{eq:raiser-powers} follows by induction.
\end{proof}

For $r>0$, define the $r$-analytic seed space
\begin{equation}
\mathcal A_r(D)
:=\left\{x\in\bigcap_{n\ge0}\mathcal D(D^n):
\sum_{n=0}^\infty
\frac{r^{2n}}{(n!)^2}\norm{D^nx}^2<\infty\right\}.
\label{eq:analytic-seed-space}
\end{equation}

\begin{theorem}[Unbounded analytic infinite-jet theorem]
If $x\in\mathcal A_r(D)$, then for every $t$ with $|t|<r$ the series
\begin{equation}
\Psi_x(t)
:=\sum_{n=0}^\infty\frac{t^n}{n!}
\mathbb R_D^n\iota_0x
=\left(\frac{t^n}{n!}D^nx\right)_{n\ge0}
\label{eq:analytic-infinite-jet}
\end{equation}
converges in $\cJ_\infty$.  The map $t\mapsto\Psi_x(t)$ is analytic, belongs to $\mathcal D(\mathbb R_D)$ for $|t|<r$, and is the unique $\cJ_\infty$-analytic solution of
\begin{equation}
\frac{\dd}{\dd t}\Psi_x(t)
=\mathbb R_D\Psi_x(t),
\qquad
\Psi_x(0)=\iota_0x.
\label{eq:infinite-jet-evolution}
\end{equation}
If $x\in\bigcap_{r>0}\mathcal A_r(D)$, the solution exists for every real or complex $t$.
\end{theorem}

\begin{proof}
Equation \eqref{eq:raiser-powers} gives the displayed layer formula, and the norm identity
\[
\norm{\Psi_x(t)}^2
=\sum_{n=0}^\infty
\frac{|t|^{2n}}{(n!)^2}\norm{D^nx}^2
\]
proves convergence for $|t|<r$.  Fix $\rho$ with $|t|<\rho<r$.  The derivative and raised series are controlled by the same analytic-seed sum because
\[
n^2\left(\frac{\rho}{r}\right)^{2n}
\]
is bounded in $n$.  Hence termwise differentiation is valid in $\cJ_\infty$, $\Psi_x(t)\in\mathcal D(\mathbb R_D)$, and the component formula gives \eqref{eq:infinite-jet-evolution}.  If $Y(t)=\sum_{n\ge0}t^nY_n$ is another analytic solution with the same initial value, coefficient comparison gives
\[
(n+1)Y_{n+1}=\mathbb R_DY_n,
\]
so $Y_n=\mathbb R_D^n\iota_0x/n!$ and $Y=\Psi_x$.  Entire analytic seeds give convergence for every $t$.
\end{proof}

\begin{corollary}[Commuting transported infinite jet]
Let $U_s$ be a strongly continuous unitary group on $H$ such that $U_s\mathcal D(D)\subseteq\mathcal D(D)$ and $DU_sx=U_sDx$ for $x\in\mathcal D(D)$.  Let $\mathbb U_s$ act diagonally on $\cJ_\infty$.  Then
\begin{equation}
\mathbb U_s\Psi_x(t)
=\left(\frac{t^n}{n!}U_sD^nx\right)_{n\ge0}
\label{eq:transported-infinite-jet}
\end{equation}
is the exact transported infinite jet, and $\mathbb U_s\mathbb R_D=\mathbb R_D\mathbb U_s$ on $\mathcal D(\mathbb R_D)$.
\end{corollary}

\begin{proof}
The commutation hypothesis iterates to $D^nU_sx=U_sD^nx$.  Applying the diagonal unitary to \eqref{eq:analytic-infinite-jet} gives \eqref{eq:transported-infinite-jet}; componentwise evaluation proves commutation with the raiser.
\end{proof}

\begin{corollary}[Thermodynamic analytic-jet realization]
Let $D=\nabla_X$ be a closed covariant derivative along a declared thermodynamic process field on a Hilbert completion of response sections.  Every response seed $\Lambda\in\mathcal A_r(\nabla_X)$ possesses the rigorous infinite response jet
\[
\left(\frac{t^n}{n!}\nabla_X^n\Lambda\right)_{n\ge0}
\in\cJ_\infty,
\qquad |t|<r.
\]
Thus the finite capstone extends to a closed unbounded infinite carrier on the analytic-seed domain.  What remains separate is generation of a global $C_0$-semigroup on the entire carrier without an analytic-vector restriction.
\end{corollary}

%% file: sections/05a_curvature_and_analytic_relevance.tex
\subsection{Typing Recognition curvature against equilibrium thermodynamic curvature}
\label{sec:curvature-typing}

Several inequivalent curvatures occur in thermodynamics and field theory.  They must not be identified merely because the same word is used.

\begin{definition}[Three curvature types]
The paper distinguishes:
\begin{enumerate}[label=(\roman*)]
\item an equilibrium metric curvature $\mathcal R_{\mathrm{eq}}$, such as the scalar or tensor curvature of a declared Weinhold or Ruppeiner metric on the equilibrium state manifold;
\item a response-transport curvature $F^{\mathrm{open}}=\dd(\omega-\ell)$ or its vector-bundle analogue, measuring the failure of a declared response transport and ledger to close;
\item the spacetime quotient curvature $F_A$ induced by the target-faithful transport square on $E_{\RK}$.
\end{enumerate}
\end{definition}

These objects live in different bundles and have different dimensions.  An equilibrium scalar curvature cannot be inserted into a force--flux law or gauge action without an explicit bundle map and unit calibration.

\begin{proposition}[Typed curvature calibration]
Let
\[
\mathfrak F
=
\mathfrak F_{\mathrm{open}}
\oplus
\mathfrak F_{\mathrm{eq}}
\oplus
\mathfrak F_{\mathrm{sp}}
\]
be a direct sum of the declared open-response, equilibrium-metric, and spacetime quotient-curvature carriers after lawful nondimensionalization.  Any constitutive production calibration
\[
\mathcal C:\mathfrak F\to\operatorname{Hom}(\R^m,\R^r)
\]
decomposes uniquely as
\[
\mathcal C(F_{\mathrm{open}},R_{\mathrm{eq}},F_A)
=
\mathcal C_{\mathrm{open}}F_{\mathrm{open}}
+
\mathcal C_{\mathrm{eq}}R_{\mathrm{eq}}
+
\mathcal C_{\mathrm{sp}}F_A
\]
when $\mathcal C$ is linear.  The positive production theorem then gives
\[
\sigma(X)
=
\norm{
\mathcal C_{\mathrm{open}}F_{\mathrm{open}}X
+
\mathcal C_{\mathrm{eq}}R_{\mathrm{eq}}X
+
\mathcal C_{\mathrm{sp}}F_AX
}^2
\ge0.
\]
No individual curvature contribution is selected by geometry alone.
\end{proposition}

\begin{proof}
The decomposition is the universal property of a direct sum.  The production identity is the calibrated curvature-to-production theorem applied to the combined typed curvature datum.
\end{proof}

\begin{remark}[Relation to Ruppeiner geometry]
Ruppeiner curvature may be included through $\mathcal C_{\mathrm{eq}}$ when an equilibrium fluctuation metric and its physical calibration are declared.  The Recognition open curvature is not asserted to equal Ruppeiner curvature.  Their equality, proportionality, or coupling would be an additional constitutive theorem to be tested in a specified model.
\end{remark}

\subsection{When thermodynamic response seeds are analytic}
\label{sec:analytic-seed-relevance}

The analytic-vector assumption in the infinite-jet theorem is restrictive but not empty.  It holds locally for analytic constitutive laws and analytic process curves away from phase singularities, and it is automatic for bounded response generators.

\begin{theorem}[Holomorphic constitutive criterion for analytic response seeds]
\label{thm:holomorphic-response-seed}
Let $H$ be a complex Hilbert response fibre and let
\[
\lambda:\{z\in\C:|z|<R\}\to H
\]
be holomorphic with
\[
\sup_{|z|<R}\norm{\lambda(z)}_H\le M.
\]
Let $D^n\Lambda:=\lambda^{(n)}(0)$.  Then for every $0<r<R$,
\begin{equation}
\sum_{n=0}^\infty
\frac{r^{2n}}{(n!)^2}
\norm{D^n\Lambda}_H^2
\le
\frac{M^2}{1-(r/R)^2}.
\label{eq:holomorphic-analytic-bound}
\end{equation}
Hence $\Lambda\in\mathcal A_r(D)$ and the infinite response jet exists for $|t|<r$.
\end{theorem}

\begin{proof}
The Hilbert-valued Cauchy estimate gives
\[
\norm{\lambda^{(n)}(0)}_H
\le
\frac{Mn!}{R^n}.
\]
Substitution yields
\[
\sum_{n=0}^\infty
\frac{r^{2n}}{(n!)^2}
\norm{D^n\Lambda}_H^2
\le
M^2\sum_{n=0}^\infty\left(\frac rR\right)^{2n},
\]
which is \eqref{eq:holomorphic-analytic-bound}.
\end{proof}

\begin{corollary}[Analytic constitutive thermodynamic processes]
Suppose the fundamental relation $U(S,V)$ is real analytic on a stable chart and a real-analytic process curve
\[
t\longmapsto(S(t),V(t))
\]
remains in a compact subset on which all denominators in the response formulas are nonzero.  Then each response coordinate and the nondimensional response section $\widehat\Lambda(t)$ are real analytic in $t$.  After shrinking the interval, they admit a bounded holomorphic extension and satisfy the analytic-seed estimate of Theorem~\ref{thm:holomorphic-response-seed}.
\end{corollary}

\begin{proof}
Derivatives, algebraic combinations, and reciprocals of nonvanishing real-analytic functions are real analytic.  Composition with a real-analytic curve preserves analyticity.  Local complexification supplies a holomorphic extension, and compactness after shrinking supplies a finite bound.
\end{proof}

\begin{corollary}[Ideal-gas and van der Waals analytic sectors]
The ideal-gas response coordinates are analytic wherever $T>0$ and $V>0$.  The van der Waals response coordinates are analytic on every compact stable region satisfying
\[
V>b,
\qquad
\left(\frac{\partial P}{\partial V}\right)_T<0,
\]
and avoiding zeros of the heat capacities and bulk moduli.  Along every analytic process curve contained in such a region, the infinite response-jet theorem applies locally.
\end{corollary}

\begin{proposition}[Bounded-generator criterion]
If $D\in\mathcal B(H)$ is bounded, then every $x\in H$ is an entire analytic seed:
\begin{equation}
\sum_{n=0}^\infty
\frac{r^{2n}}{(n!)^2}\norm{D^nx}^2
\le
\norm x^2
\sum_{n=0}^\infty
\frac{(r\norm D)^{2n}}{(n!)^2}
<\infty
\label{eq:bounded-generator-entire}
\end{equation}
for every $r>0$.
\end{proposition}

\begin{proof}
Use $\norm{D^nx}\le\norm D^n\norm x$ and convergence of the factorially damped scalar series.
\end{proof}

\begin{remark}[Failure at singular regimes]
Near a spinodal, critical point, phase boundary, vanishing heat capacity, or any chart degeneration, the analytic radius may collapse and the response generator may cease to satisfy the required estimate.  The theorem then does not fail mysteriously; its domain certificate reports that the chosen infinite-jet representation is no longer lawful.
\end{remark}

%% file: sections/05b_extended_capstone.tex
\section{Extended Thermodynamic Recognition Capstone}
\label{sec:extended-capstone}

The previous results can be consumed as one strengthened theorem without confusing their hypotheses.

\begin{theorem}[Extended Thermodynamic Recognition Capstone]
Assume the hypotheses of the finite Thermodynamic Recognition Capstone.  Assume in addition:
\begin{enumerate}[label=(\roman*)]
\item the equilibrium chart carries a declared nowhere-vanishing area form $\varpi$ and four selected response channels have been nondimensionalized to a smooth map $F:\cM\to\R^4$;
\item an open-curvature space $\mathfrak F$, force space $\R^m$, and prior calibration
\[
\mathcal C:\mathfrak F\to\operatorname{Hom}(\R^m,\R^r)
\]
have been declared;
\item the covariant process derivative $D=\nabla_X$ is densely defined and closed on a Hilbert completion $H$ of response sections;
\item the response seed $\Lambda$ belongs to $\mathcal A_r(D)$ for some $r>0$.
\end{enumerate}
Then all finite-capstone conclusions remain valid and the following strengthened conclusions hold:
\begin{enumerate}[label=(\alph*)]
\item the response-bracket matrix
\[
\mathbb B_\varpi(F)_{ab}=\{f_a,f_b\}_\varpi
\]
is skew, has rank at most two, transforms by congruence under constant response-basis changes, and obeys
\[
B_{12}B_{34}-B_{13}B_{24}+B_{14}B_{23}=0;
\]
\item for every admissible curvature value $F^{\mathrm{open}}$ and force $X$, the calibrated response law
\[
J=\bigl(\mathcal C(F^{\mathrm{open}})^{\mathsf T}
\mathcal C(F^{\mathrm{open}})+A(F^{\mathrm{open}})\bigr)X,
\qquad A^{\mathsf T}=-A,
\]
has nonnegative production
\[
J^{\mathsf T}X
=\norm{\mathcal C(F^{\mathrm{open}})X}^2\ge0;
\]
\item the completed carrier $\cJ_\infty=\ell^2(\N_0;H)$ supports the closed unbounded raiser $\mathbb R_D$, and for $|t|<r$ the exact infinite response jet is
\[
\Psi_\Lambda(t)
=\left(\frac{t^n}{n!}D^n\Lambda\right)_{n\ge0};
\]
\item if $\Lambda$ is entire analytic for $D$, the infinite-jet orbit exists for every $t$; if a unitary process transport commutes strongly with $D$, it transports the tower layer by layer without changing the factorial coefficients.
\end{enumerate}
No finite certificate is used to infer these infinite or constitutive conclusions.
\end{theorem}

\begin{proof}
The finite conclusions are those of the Thermodynamic Recognition Capstone.  Part (a) is the universal skew four-channel bracket theorem.  Part (b) is the calibrated curvature-to-production theorem.  Part (c) is the closedness lemma and unbounded analytic infinite-jet theorem applied to $D=\nabla_X$.  Part (d) follows from the entire-seed statement and the commuting transported-jet corollary.  Each conclusion consumes only the additional datum named in its corresponding hypothesis.
\end{proof}

\begin{remark}[Expanded but still typed]
The theorem genuinely enlarges the claim boundary: it supplies an intrinsic skew four-channel object, a positive curvature-generated production functional, and an unbounded infinite response tower.  It does not identify the ordinary derivative Jacobian with the bracket matrix, choose a physical calibration from geometry alone, or promote an analytic-vector orbit into a whole-carrier semigroup without a generator theorem.
\end{remark}

%% file: sections/05c_novelty_falsifiability.tex
\section{Logical dependencies, novelty, and falsifiable signatures}
\label{sec:dependency-falsifiability}

A broad theorem package is credible only when its dependency arrows are visible.  This section records which conclusions are intrinsic to equilibrium and Recognition geometry, which require additional dynamics, and which admit direct experimental rejection tests.

\subsection{Dependency ladder}

The logical layers used in this paper are:
\begin{center}
\begin{tabular}{p{0.10\textwidth}p{0.32\textwidth}p{0.46\textwidth}}
\toprule
Layer & Additional datum & Conclusion supplied\\
\midrule
$\mathsf L_0$ & Fundamental relation $U(S,V)$ and regular equilibrium chart & Legendre structure, Maxwell relations, constrained response identities\\
$\mathsf L_1$ & Observation, target, lawful ledger & Recognition kernel, target-faithful completion, repair rank, decoder burden, loop residue\\
$\mathsf L_2$ & Connection, graded jet carrier, self-adjoint cut & Finite and analytic response jets, cut sectors, commutator curvature, zero-mode projection\\
$\mathsf L_3$ & Local-equilibrium map, constant-rank observation/target bundle maps, compatible transport & Thermodynamic Recognition quotient bundle, explicit field $\Phi_{\mathrm{th}}$, target-faithful transport square, induced connection and curvature, uniform source domination\\
$\mathsf L_4$ & Lorentzian metric plus locality, symmetry, parity, and derivative-order contract & Minimal invariant Einstein--Hilbert plus quotient-gauge--field action class and its residual constitutive freedom\\
$\mathsf L_5$ & Couplings and invariant potential within the classified action & Metric, gauge, and Recognition-field Euler--Lagrange equations\\
$\mathsf L_6$ & Gauge-fixed coercive finite regulator & Exact finite partition function, Schwinger--Dyson and effective equations\\
$\mathsf L_7$ & Renormalized anomaly-free continuum effective action & Quantum-effective metric, gauge, matter, and Ward equations\\
$\mathsf L_8$ & Positive target-faithful physical Hessian and nuclear spectral test space & Continuum Gaussian measure, regulator removal, reflection positivity, unitary reconstruction, bosonic anomaly freedom\\
$\mathsf L_9$ & Exact propagation, calibrated scale, and complete parallel transport & Metric and connection identification up to diffeomorphism and gauge\\
$\mathsf L_{10}$ & Transverse slice, stable linearized decoder, derivative control, uncertainty bounds & Noise-stable local reconstruction and finite-error falsification threshold\\
\bottomrule
\end{tabular}
\end{center}

No conclusion is promoted downward through this table.  In particular, $\mathsf L_0$--$\mathsf L_2$ do not by themselves produce a spacetime field.  Layer $\mathsf L_3$ constructs that field and its connection from pulled-back thermodynamic response jets.  Layer $\mathsf L_4$ classifies, but does not numerically calibrate, the minimal action inside an explicit local symmetry class.  Layer $\mathsf L_6$ does not imply nonlinear regulator removal, and the constructive $\mathsf L_8$ theorem does not assert nonlinear ultraviolet completion outside its positive quadratic sector.  The purpose of the ladder is not bureaucratic caution.  It prevents a lawful result at one layer from being mistaken for an existence or selection theorem at the next.

\subsection{Novelty map}

Several ingredients used later are established mathematical structures: Legendre thermodynamics, Stokes' theorem, quotient vector bundles, metric connections, Einstein--Hilbert variation, Yang--Mills variation, finite-dimensional Schwinger--Dyson integration by parts, Gaussian nuclear-space measures, and Osterwalder--Schrader reconstruction.  They are rederived where needed because their hypotheses and interfaces are load bearing, not because the paper claims to have rediscovered them.

The framework-specific contributions are the target-relative blind quotient and canonical minimal completion; the exact repair-rank and decoder-burden statements; their thermodynamic finite-jet specialization; the intrinsic four-response bracket and its Pluecker constraint; the lawful ledger/open-residue separation; and the cut--flow--jet synthesis with analytic infinite-jet extension.  The spacetime bridge is also framework-specific: a local-equilibrium map pulls the response jets to spacetime; the target-relevant kernel quotient becomes the Recognition field bundle; compatible response transport descends to a connection; its curvature is the infinitesimal defect of the target-faithful transport square; and a bundle decoder gives uniform source domination and strict target no-blindness.  The action is not selected by a resemblance between curvature symbols.  It is classified within an explicit parity-even local low-derivative invariant class, while couplings and nonlinear potential remain separately typed.  Finally, the fixed-decoder spacetime theorem turns exact geometric identifiability into a finite-error reconstruction estimate.  Thus the novelty lies in the Recognition architecture and its exact interfaces.  Standard variational and constructive equations are consequences of the constructed carrier plus declared dynamical and analytic data.

\subsection{Noise-stable Pluecker falsification test}

The rank-two response-bracket theorem gives an observable algebraic constraint.  The following estimate turns it into a finite-error test.

At a fixed state choose a local oriented coframe $(\alpha,\beta)$ with $\varpi=\alpha\wedge\beta$ and write
\[
\dd f_a=u_a\alpha+v_a\beta,
\qquad
u=(u_a)_{a=1}^4,
\qquad
v=(v_a)_{a=1}^4.
\]
To prevent the vector $u$ from being confused with the Greek letter $\nu$, write explicitly
\[
u:=(u_a)_{a=1}^4.
\]
Then
\[
B=uv^{\mathsf T}-vu^{\mathsf T}.
\]
Suppose estimated gradients give
\[
\widetilde u=u+\delta u,
\qquad
\widetilde v=v+\delta v,
\qquad
\norm{\delta u}\le\varepsilon_u,
\qquad
\norm{\delta v}\le\varepsilon_v,
\]
and set
\[
\widetilde B=\widetilde u\widetilde v^{\mathsf T}
-\widetilde v\widetilde u^{\mathsf T}.
\]
For a skew matrix $C$, let
\[
\beta(C)=(C_{12},C_{13},C_{14},C_{23},C_{24},C_{34})\in\R^6
\]
and define its Pluecker residual
\[
\mathfrak P(C)
=C_{12}C_{34}-C_{13}C_{24}+C_{14}C_{23}.
\]

\begin{theorem}[Noise-stable response-bracket test]
\label{thm:noise-stable-pluecker}
Define
\begin{equation}
E
:=2\left(
\varepsilon_u\norm{\widetilde v}
+\varepsilon_v\norm{\widetilde u}
+3\varepsilon_u\varepsilon_v
\right),
\qquad
\eta:=\frac{E}{\sqrt2}.
\label{eq:measured-bracket-error}
\end{equation}
Then
\begin{align}
\norm{\widetilde B-B}_{\mathrm F}&\le E,
\label{eq:bracket-frobenius-error}\\
\norm{\beta(\widetilde B)-\beta(B)}&\le\eta,
\label{eq:bracket-six-error}
\end{align}
and, because $\mathfrak P(B)=0$,
\begin{equation}
\abs{\mathfrak P(\widetilde B)}
\le
\norm{\beta(\widetilde B)}\,\eta
+\frac32\eta^2.
\label{eq:pluecker-noise-bound}
\end{equation}
Consequently, a measured residual larger than the right-hand side of \eqref{eq:pluecker-noise-bound} rejects the conjunction of the declared two-dimensional equilibrium chart, channel assignment, nondimensionalization, and error bounds at that state.
\end{theorem}

\begin{proof}
Expanding the measured matrix gives
\[
\widetilde B-B
=\delta u\,v^{\mathsf T}+u\,\delta v^{\mathsf T}
-\delta v\,u^{\mathsf T}-v\,\delta u^{\mathsf T}
+\delta u\,\delta v^{\mathsf T}
-\delta v\,\delta u^{\mathsf T}.
\]
For rank-one matrices, $\norm{ab^{\mathsf T}}_{\mathrm F}=\norm a\norm b$.  Hence
\[
\norm{\widetilde B-B}_{\mathrm F}
\le2\left(
\varepsilon_u\norm v
+\varepsilon_v\norm u
+\varepsilon_u\varepsilon_v
\right).
\]
Using $\norm u\le\norm{\widetilde u}+\varepsilon_u$ and
$\norm v\le\norm{\widetilde v}+\varepsilon_v$ gives \eqref{eq:bracket-frobenius-error}.  Since the difference is skew,
\[
\norm{\widetilde B-B}_{\mathrm F}^2
=2\norm{\beta(\widetilde B)-\beta(B)}^2,
\]
which proves \eqref{eq:bracket-six-error}.

Let $b=\beta(B)$ and $e=\beta(\widetilde B)-b$.  There is a symmetric orthogonal matrix $J_6$ pairing the coordinates $(1,6)$, $(2,5)$ and $(3,4)$ with signs $+,-,+$ such that
\[
\mathfrak P(B)=\frac12b^{\mathsf T}J_6b.
\]
Therefore
\[
\mathfrak P(\widetilde B)-\mathfrak P(B)
=b^{\mathsf T}J_6e+\frac12e^{\mathsf T}J_6e,
\]
so
\[
\abs{\mathfrak P(\widetilde B)}
\le\norm b\norm e+\frac12\norm e^2.
\]
Finally, $\norm b\le\norm{\beta(\widetilde B)}+\norm e$ and $\norm e\le\eta$, which yields \eqref{eq:pluecker-noise-bound}.
\end{proof}

\begin{remark}[What a failed test means]
The theorem is a falsification test for the full declared measurement contract, not a detector of one preferred failure mechanism.  A violation can indicate departure from a smooth two-dimensional equilibrium regime, a hidden internal variable, incorrect constraints, unit or channel mismatch, or underestimated error.  The Recognition ledger then records which assumption failed instead of retroactively renaming the residual.
\end{remark}

%% file: sections/06a0_thermodynamic_recognition_square.tex
\section{From thermodynamic response jets to the spacetime Recognition bundle}
\label{sec:thermodynamic-recognition-square}

The spacetime fields used later are now constructed from the thermodynamic Recognition data rather than introduced as unrelated variables.  The construction is the bundle form of the quotient-adapter, compatibility-transfer, source-domination, and no-blindness principles of the Recognition Kernel framework \cite{DabasRKF}.

\subsection{Local-equilibrium pullback and the target-relevant bundle}

Let $\cM_{\mathrm{th}}$ be the admissible thermodynamic state manifold and let $E_\Lambda\to\cM_{\mathrm{th}}$ be the direct-sum unit bundle of the six response coordinates.  A smooth local-equilibrium state field on spacetime is a map
\[
\Theta:M\longrightarrow\cM_{\mathrm{th}}.
\]
For a fixed finite jet order $N$, define the pulled-back thermodynamic response-jet bundle
\begin{equation}
\cH^N_\Theta
:=\Theta^*J^N(E_\Lambda)
\longrightarrow M.
\label{eq:thermodynamic-pullback-jet}
\end{equation}
Its distinguished response section is
\[
\mathbf j^N_\Theta
:=\Theta^*(j^N\Lambda)\in\Gamma(\cH^N_\Theta).
\]

Let $\cY\to M$ and $\cZ\to M$ be metric vector bundles and let
\[
\cA:\cH^N_\Theta\to\cY,
\qquad
\Pi:\cH^N_\Theta\to\cZ
\]
be smooth bundle morphisms representing the local observation protocol and the declared target.  Assume that $\cA$ and the stacked map $(\cA,\Pi)$ have constant rank.  Set
\begin{equation}
\cN:=\Ker\cA,
\qquad
\cN_0:=\Ker\cA\cap\Ker\Pi
=\Ker(\cA,\Pi).
\label{eq:bundle-blind-subspaces}
\end{equation}
The constant-rank hypothesis makes $\cN$ and $\cN_0$ smooth subbundles.

\begin{theorem}[Thermodynamic Recognition quotient-bundle theorem]
\label{thm:thermodynamic-recognition-bundle}
Under the hypotheses above, the fibrewise target-relevant blind quotient
\begin{equation}
E_{\RK}
:=\cB(\cA;\Pi)
:=\cN/\cN_0
\longrightarrow M
\label{eq:recognition-quotient-bundle}
\end{equation}
is a smooth metric vector bundle.  Let $P_{\cN}$ be the orthogonal projection onto $\cN$ and let $q:\cN\to E_{\RK}$ be the quotient map.  The canonical completed observation
\begin{equation}
\cA_\Pi^\sharp v
:=\bigl(\cA v,qP_{\cN}v\bigr)
:\cH^N_\Theta\longrightarrow\cY\oplus E_{\RK}
\label{eq:bundle-canonical-completion}
\end{equation}
is a smooth bundle morphism satisfying
\begin{equation}
\Ker\cA_\Pi^\sharp=\cN_0.
\label{eq:bundle-completion-kernel}
\end{equation}
The thermodynamic Recognition field is the explicit quotient section
\begin{equation}
\Phi_{\mathrm{th}}
:=qP_{\cN}\mathbf j^N_\Theta
\in\Gamma(E_{\RK}).
\label{eq:thermodynamic-recognition-field}
\end{equation}
Thus $\Phi_{\mathrm{th}}(x)$ is precisely the target-relevant part of the local thermodynamic response jet that is invisible to the original observation at $x$.
\end{theorem}

\begin{proof}
Constant rank of $\cA$ gives the smooth kernel subbundle $\cN$, while constant rank of $(\cA,\Pi)$ gives the smooth subbundle $\cN_0$.  The quotient of a finite-rank vector bundle by a smooth subbundle is a smooth vector bundle.  The metric on $\cH^N_\Theta$ identifies $\cN/\cN_0$ with the orthogonal complement $\cN\cap\cN_0^\perp$, providing a quotient metric and showing that $P_{\cN}$ and $q$ vary smoothly.

If $\cA_\Pi^\sharp v=0$, then $\cA v=0$, hence $v\in\cN$ and $P_{\cN}v=v$.  The second component gives $qv=0$, so $v\in\cN_0$.  Conversely every $v\in\cN_0$ annihilates both components.  This proves \eqref{eq:bundle-completion-kernel}.  Applying the smooth morphism $qP_{\cN}$ to the pulled-back response jet gives \eqref{eq:thermodynamic-recognition-field}.
\end{proof}

\begin{corollary}[Target meaning of the spacetime Recognition field]
\label{cor:recognition-field-meaning}
The field $\Phi_{\mathrm{th}}$ vanishes exactly when the thermodynamic response jet has no target-relevant component hidden from the declared observation.  It is invariant under every unitary change of response representation that transports $\cA$, $\Pi$, and the fibre metric covariantly.
\end{corollary}

\subsection{The target-faithful transport square}

Let $\nabla^{\cH}$, $\nabla^{\cY}$, and $\nabla^{\cZ}$ be metric connections on $\cH^N_\Theta$, $\cY$, and $\cZ$.  Assume that observation and target are parallel bundle morphisms:
\begin{align}
\nabla^{\cY}_X(\cA v)&=\cA(\nabla^{\cH}_Xv),
\label{eq:observation-parallel}\\
\nabla^{\cZ}_X(\Pi v)&=\Pi(\nabla^{\cH}_Xv)
\label{eq:target-parallel}
\end{align}
for every vector field $X$ and smooth section $v$.  These identities imply that $\nabla^{\cH}$ preserves $\cN$ and $\cN_0$.

\begin{theorem}[Target-faithful transport-square theorem]
\label{thm:target-faithful-square}
The connection $\nabla^{\cH}$ descends uniquely to a metric connection $\nabla^{\RK}$ on $E_{\RK}$ by
\begin{equation}
\nabla_X^{\RK}[v]
:=[\nabla_X^{\cH}v],
\qquad v\in\Gamma(\cN).
\label{eq:quotient-connection}
\end{equation}
Let $U_\gamma^{\cH}$, $U_\gamma^{\cY}$, and $U_\gamma^{\RK}$ denote parallel transport along a piecewise smooth path $\gamma:x\to y$.  Then
\begin{equation}
\cA_{\Pi,y}^\sharp U_\gamma^{\cH}
=
\bigl(U_\gamma^{\cY}\oplus U_\gamma^{\RK}\bigr)
\cA_{\Pi,x}^\sharp.
\label{eq:target-faithful-square}
\end{equation}
Equivalently, the square
\[
\begin{array}{ccc}
(\cH^N_\Theta)_x & \xrightarrow{\ U_\gamma^{\cH}\ } & (\cH^N_\Theta)_y\\[2mm]
\Big\downarrow{\cA_{\Pi,x}^\sharp} &&
\Big\downarrow{\cA_{\Pi,y}^\sharp}\\[2mm]
\cY_x\oplus(E_{\RK})_x &
\xrightarrow{\ U_\gamma^{\cY}\oplus U_\gamma^{\RK}\ } &
\cY_y\oplus(E_{\RK})_y
\end{array}
\]
commutes.  The transported Recognition field is therefore the quotient transport of the target-relevant thermodynamic response residue, not an independently named scalar field.
\end{theorem}

\begin{proof}
If $v$ and $v+w$ represent the same quotient section, then $w\in\Gamma(\cN_0)$.  Preservation of $\cN_0$ implies $\nabla_X^{\cH}w\in\Gamma(\cN_0)$, so \eqref{eq:quotient-connection} is well defined.  Metric compatibility descends from the orthogonal realization $E_{\RK}\simeq\cN\cap\cN_0^\perp$.

Equations \eqref{eq:observation-parallel}--\eqref{eq:target-parallel} integrate to the corresponding parallel-transport intertwining relations.  In particular, $U_\gamma^{\cH}$ maps $\cN_x$ unitarily onto $\cN_y$ and $\cN_{0,x}$ onto $\cN_{0,y}$, hence it induces $U_\gamma^{\RK}$.  Unitary transport also intertwines the orthogonal projections onto $\cN_x$ and $\cN_y$.  Applying these identities to the two components of \eqref{eq:bundle-canonical-completion} proves \eqref{eq:target-faithful-square}.
\end{proof}

\begin{theorem}[Recognition curvature as the infinitesimal square defect]
\label{thm:recognition-curvature-square-defect}
Let $F^{\cH}$ and $F^{\RK}$ be the curvatures of $\nabla^{\cH}$ and $\nabla^{\RK}$.  For $v\in\cN$,
\begin{equation}
F^{\RK}(X,Y)[v]
=[F^{\cH}(X,Y)v].
\label{eq:quotient-curvature}
\end{equation}
For an oriented infinitesimal rectangle $\square_\varepsilon(X,Y)$ based at $x$, the induced quotient holonomy satisfies
\begin{equation}
U^{\RK}_{\partial\square_\varepsilon(X,Y)}
=\id+\varepsilon^2F^{\RK}_x(X,Y)+o(\varepsilon^2).
\label{eq:small-loop-square-defect}
\end{equation}
Consequently, $F^{\RK}$ is exactly the leading failure of path-independent target-faithful thermodynamic reconstruction.  If $P\to M$ is the orthonormal frame bundle of $E_{\RK}$, the local connection form $A$ of $\nabla^{\RK}$ and its curvature $F_A$ are the spacetime Recognition connection and curvature used in the action below.
\end{theorem}

\begin{proof}
Apply the definition of curvature to \eqref{eq:quotient-connection}.  Every commutator and Lie-bracket term passes to the quotient because $\cN_0$ is preserved, giving \eqref{eq:quotient-curvature}.  The small-loop expansion is the standard curvature expansion of parallel transport with the orientation convention fixed by $(X,Y)$.  The final statement is the usual equivalence between a metric vector-bundle connection and a principal connection on its orthonormal frame bundle.
\end{proof}

\subsection{Uniform source domination and strict target no-blindness}

\begin{theorem}[Bundle decoder and source-domination theorem]
\label{thm:bundle-source-domination}
Assume the preceding constant-rank hypotheses.  There exists a unique smooth bundle morphism
\[
D_{\min}:\Ran\cA_\Pi^\sharp\longrightarrow\cZ
\]
with image-orthogonal minimum norm such that
\begin{equation}
\Pi=D_{\min}\cA_\Pi^\sharp.
\label{eq:bundle-decoder-factorization}
\end{equation}
After extending $D_{\min}$ by zero on $(\Ran\cA_\Pi^\sharp)^\perp$, it is a smooth morphism $\cY\oplus E_{\RK}\to\cZ$.  Fibrewise,
\begin{equation}
\norm{D_{\min,x}}^2
=\beta_{\cA_{\Pi,x}^\sharp}(\Pi_x),
\label{eq:bundle-decoder-burden}
\end{equation}
where the right-hand side denotes the operator-valued analogue of the sharp decoder burden.  If
\[
\sup_{x\in M}\norm{D_{\min,x}}\le C<\infty,
\]
then the uniform source-domination estimate
\begin{equation}
\norm{\Pi_xv}
\le C\norm{\cA_{\Pi,x}^\sharp v}
\label{eq:uniform-source-domination}
\end{equation}
holds for every $x$ and $v$.  Hence no target-relevant mode is blind after canonical completion.  When the bundle maps and connections are parallel, $D_{\min}$ intertwines parallel transport.
\end{theorem}

\begin{proof}
Equation \eqref{eq:bundle-completion-kernel} gives
$\Ker\cA_\Pi^\sharp\subseteq\Ker\Pi$.  Therefore
$D_0(\cA_\Pi^\sharp v):=\Pi v$ is well defined on the image.  Constant rank makes the image a smooth subbundle.  On each fibre, the minimum-norm decoder is obtained by restricting to $(\Ker\cA_\Pi^\sharp)^\perp$ and inverting $\cA_\Pi^\sharp$ onto its range; these inverses vary smoothly under constant rank.  This proves the smooth factorization and its uniqueness in the image-orthogonal class.  The operator norm is the sharp factorization constant, giving \eqref{eq:bundle-decoder-burden}; the uniform estimate follows immediately.  Parallelness of all maps makes both sides of \eqref{eq:bundle-decoder-factorization} transport-covariant, and uniqueness gives the intertwining property.
\end{proof}

\subsection{Why the minimal spacetime action has the stated form}

The preceding theorems construct the fibre, field, connection, and curvature from thermodynamic response data.  They do not determine coupling constants or a nonlinear potential.  They do, however, sharply constrain the lowest-order local action once the usual locality and symmetry contract is declared.

\begin{theorem}[Minimal invariant Recognition-action classification]
\label{thm:minimal-action-classification}
Let $E_{\RK}\to(M,g)$ be the metric quotient bundle from Theorem~\ref{thm:thermodynamic-recognition-bundle}, let $A$ be the induced metric connection, and let $\Phi=\Phi_{\mathrm{th}}$ or a dynamical section of the same bundle.  Consider real local Lagrangian densities that:
\begin{enumerate}[label=(\roman*)]
\item are invariant under diffeomorphisms and the orthogonal or unitary gauge group of $E_{\RK}$;
\item are parity even and use no orientation tensor;
\item depend on the metric curvature only linearly through complete metric contractions;
\item are quadratic in $F_A$ and in $D_A\Phi$, with no mixed or nonminimal curvature couplings;
\item contain no derivative-free field term except a smooth gauge-invariant potential.
\end{enumerate}
Then, up to a boundary term and constant normalizations, every such density is
\begin{equation}
\mathcal L
=
 aR_g+b
-c\langle F_{A\,\mu\nu},F_A^{\mu\nu}\rangle
-d\langle D_{A\,\mu}\Phi,D_A^\mu\Phi\rangle
-V(\Phi),
\label{eq:minimal-action-classification}
\end{equation}
where $a,b,c,d\in\R$ and $V$ is gauge invariant.  Positivity of the Euclidean quadratic gauge and field sectors requires $c,d>0$.  With
\[
a=\frac1{2\kappa},\qquad
b=-\frac{\Lambda_{\mathrm{cos}}}{\kappa},\qquad
c=\frac1{4g_{\RK}^2},\qquad
d=\frac12,
\]
Equation~\eqref{eq:minimal-action-classification} is exactly the action density used in Section~\ref{sec:spacetime-action}.
\end{theorem}

\begin{proof}
A complete metric contraction linear in the Riemann tensor reduces, by its algebraic symmetries, to a multiple of the scalar curvature; the derivative-free metric density is a constant multiple of $\mathrm{vol}_g$.  With the fixed invariant inner product on the adjoint bundle and parity-odd contractions excluded, the unique quadratic complete contraction of the gauge curvature is $\langle F_{\mu\nu},F^{\mu\nu}\rangle$.  With the fixed fibre metric, the unique quadratic complete contraction of one covariant derivative of $\Phi$ is $\langle D_\mu\Phi,D^\mu\Phi\rangle$.  The remaining derivative-free field contribution is a gauge-invariant function $V(\Phi)$.  Adding the independent coefficients gives \eqref{eq:minimal-action-classification}; integration by parts absorbs equivalent second-derivative representatives into a boundary term.  Euclidean positivity fixes the signs of the quadratic coefficients.
\end{proof}

\begin{remark}[Exact selection statement]
The thermodynamic Recognition square selects the \emph{meaning} of $E_{\RK}$, $\Phi$, $A$, and $F_A$.  The locality, symmetry, parity, and derivative-order contract selects the minimal action form within the stated class.  Neither theorem fixes the numerical couplings or the invariant potential, and higher-curvature, nonminimal, parity-odd, or nonlocal actions lie outside this classification rather than being silently excluded from existence.
\end{remark}

%% file: sections/06a00_response_propagation_metric.tex
\subsection{Recognition metric from target-faithful response propagation}
\label{sec:recognition-metric-propagation}

The quotient bundle and its connection are constructed from thermodynamic response transport.  A spacetime metric enters when the target-relevant quotient field is assigned a propagation law.  The correct geometric datum is the principal symbol of that propagation operator, not a visual analogy between thermodynamic and spacetime curvature.

Let
\[
\mathcal P_{\RK}:\Gamma(E_{\RK})\longrightarrow\Gamma(E_{\RK})
\]
be a second-order linear differential operator describing linearized propagation of the target-relevant thermodynamic Recognition field.  Assume that its principal symbol is scalar on the quotient fibre:
\begin{equation}
\sigma_2(\mathcal P_{\RK})(x,\xi)
=p_x(\xi)\id_{(E_{\RK})_x},
\label{eq:scalar-recognition-symbol}
\end{equation}
where $p_x$ is a smooth nondegenerate Lorentzian quadratic form on $T_x^*M$.

\begin{theorem}[Principal-symbol Recognition metric theorem]
\label{thm:principal-symbol-recognition-metric}
Under \eqref{eq:scalar-recognition-symbol}, there is a unique Lorentzian co-metric $g^{-1}$ satisfying
\begin{equation}
p_x(\xi)=g_x^{-1}(\xi,\xi).
\label{eq:recognition-cometric-symbol}
\end{equation}
It is reconstructed pointwise by polarization:
\begin{equation}
g_x^{-1}(\xi,\eta)
=
\frac12\bigl(
 p_x(\xi+\eta)-p_x(\xi)-p_x(\eta)
\bigr).
\label{eq:cometric-polarization}
\end{equation}
Consequently, a normalized target-faithful response propagation operator determines the spacetime metric used in the action.

If only the characteristic set
\[
\operatorname{Char}_x(\mathcal P_{\RK})
=\{\xi\ne0:p_x(\xi)=0\}
\]
is observed, it determines the conformal class of $g$.  A pointwise volume, proper-time, or proper-length calibration fixes the remaining conformal factor.
\end{theorem}

\begin{proof}
A quadratic form is uniquely determined by its polarization, which gives \eqref{eq:cometric-polarization}.  Nondegeneracy and Lorentzian signature of $p_x$ therefore define a unique Lorentzian co-metric and its inverse metric.  If only the null cone is known, two Lorentzian quadratic forms in dimension at least three with the same nonzero null set are positive scalar multiples of one another.  Hence the cone determines the conformal class.  The scale arguments are those proved in Theorem~\ref{thm:null-cone-identification}.
\end{proof}

\begin{definition}[Target-faithful propagation compatibility]
Let $\mathcal P_{\cH}$ be a second-order response-jet propagation operator on $\cH^N_\Theta$, and let $\mathcal P_{\cY}$ be the corresponding measured-data propagation.  The propagation is target-faithful when
\begin{equation}
\cA_\Pi^\sharp\mathcal P_{\cH}
=
(\mathcal P_{\cY}\oplus\mathcal P_{\RK})
\cA_\Pi^\sharp
\label{eq:propagation-intertwining-square}
\end{equation}
on a common smooth core.
\end{definition}

\begin{proposition}[No hidden target cone after completion]
\label{prop:no-hidden-target-cone}
Assume \eqref{eq:propagation-intertwining-square} and the uniform source-domination estimate \eqref{eq:uniform-source-domination}.  Then no high-frequency response-jet mode can be invisible to the completed observation while producing a nonzero target amplitude.  In particular, every target-relevant characteristic covector of $\mathcal P_{\cH}$ appears in either the measured-data symbol or the quotient-field symbol.
\end{proposition}

\begin{proof}
Let $v$ be a principal-amplitude vector at $(x,\xi)$ lying in the kernel of the completed principal observation.  Source domination gives $\Pi_xv=0$, so such a mode is target irrelevant.  Conversely, if $\Pi_xv\ne0$, then $\cA_{\Pi,x}^\sharp v\ne0$.  Applying the principal symbol of \eqref{eq:propagation-intertwining-square} shows that its propagated amplitude is represented in the measured or quotient component.  Thus a target-relevant characteristic cannot disappear into the original blind kernel after completion.
\end{proof}

\begin{remark}[Exact metric claim]
The theorem does not assert that equilibrium thermodynamics alone creates a Lorentzian metric.  It proves the sharper statement that once the target-relevant thermodynamic quotient is assigned a scalar normally hyperbolic propagation operator, the metric is encoded uniquely in its normalized principal symbol, and target-faithful completion prevents a physically relevant propagation cone from remaining observationally blind.
\end{remark}

%% file: sections/06a_spacetime_recognition_datum.tex
\section{Spacetime Recognition action}
\label{sec:spacetime-action}

Section~\ref{sec:thermodynamic-recognition-square} removes the formerly missing interface between thermodynamic response and spacetime field theory.  A local-equilibrium field $\Theta$ and the response jet $j^N\Lambda$ define the pulled-back carrier $\cH^N_\Theta$.  The observation and target maps define the quotient bundle
\[
E_{\RK}=\Ker\cA/(\Ker\cA\cap\Ker\Pi),
\]
the distinguished thermodynamic Recognition field $\Phi_{\mathrm{th}}$, and the induced quotient connection $\nabla^{\RK}$.  Its curvature is the leading target-faithful transport-square defect.  Thus the fibre, field, connection, and curvature used below are constructed Recognition objects rather than a scalar and gauge field renamed after their equations are written.

Let $(M,g)$ be an oriented, time-oriented Lorentzian manifold of dimension $n\ge3$.  The metric may be supplied as calibrated spacetime data or reconstructed from the normalized scalar principal symbol of target-faithful quotient propagation by Theorem~\ref{thm:principal-symbol-recognition-metric}; characteristic data alone determine its conformal class.  Assume either that $M$ has no boundary or that all variations below have compact support.  Let $P\to M$ be a compact orthogonal or unitary frame reduction of $E_{\RK}$ that preserves its fibre metric, target structure, and any declared cut sectors.  Write its compact structure group as $G_{\RK}$, its Lie algebra as $\mathfrak g_{\RK}$, and equip $\mathfrak g_{\RK}$ with an $\operatorname{Ad}$-invariant positive inner product $\langle\cdot,\cdot\rangle_{\mathfrak g}$.  Equivalently,
\[
E_{\RK}=P\times_\rho H_{\RK},
\]
where $H_{\RK}$ is the typical fibre of the target-relevant quotient or a declared Hilbert completion of it.

The quotient connection $\nabla^{\RK}$ is represented in a local frame by a principal connection $A$ on $P$ and induces the covariant derivative $D_A$ on $E_{\RK}$.  Its curvature
\[
F_A\in\Omega^2(M;\operatorname{ad}P)
\]
is the local representative of the quotient curvature in \eqref{eq:quotient-curvature}.  A spacetime Recognition field is a section
\[
\Phi\in\Gamma(E_{\RK}).
\]
The section $\Phi_{\mathrm{th}}$ in \eqref{eq:thermodynamic-recognition-field} is the response-generated configuration singled out by the thermodynamic data; allowing $\Phi$ to vary gives the dynamical completion of that same target-relevant carrier.  Let $V_{\RK}:E_{\RK}\to\R$ be a smooth $G_{\RK}$-invariant fibre potential.

\begin{definition}[Action-lifted Recognition datum]
An action-lifted Recognition datum is
\[
\mathfrak R_{\mathrm{act}}
=(M,g;\Theta,\cH^N_\Theta,\cA,\Pi;
 P,A;E_{\RK},\Phi;V_{\RK};
 \kappa,g_{\RK},\Lambda_{\mathrm{cos}}),
\]
where $\kappa>0$, $g_{\RK}>0$, and $\Lambda_{\mathrm{cos}}\in\R$.  Its minimal parity-even local action, classified by Theorem~\ref{thm:minimal-action-classification}, is
\begin{equation}
\begin{split}
S_{\mathrm{ARK}}[g,A,\Phi]
:=\int_M\Bigg[
&\frac{1}{2\kappa}(R_g-2\Lambda_{\mathrm{cos}})
-\frac{1}{4g_{\RK}^2}
 \langle F_{A\,\mu\nu},F_A^{\mu\nu}\rangle_{\mathfrak g}
\\
&-\frac12\langle D_{A\,\mu}\Phi,D_A^{\mu}\Phi\rangle_{E_{\RK}}
-V_{\RK}(\Phi)
\Bigg]\,\mathrm{vol}_g.
\end{split}
\label{eq:ark-action}
\end{equation}
\end{definition}

For $\xi\in\mathfrak g_{\RK}$, define the Recognition current by
\begin{equation}
\langle J_\Phi^\nu,\xi\rangle_{\mathfrak g}
:=\operatorname{Re}\langle D_A^\nu\Phi,\rho_*(\xi)\Phi\rangle_{E_{\RK}}.
\label{eq:recognition-current}
\end{equation}
Define
\begin{align}
T^{A}_{\mu\nu}
&:=\frac{1}{g_{\RK}^2}
\left(
\langle F_{A\,\mu\alpha},F_{A\,\nu}{}^{\alpha}\rangle_{\mathfrak g}
-\frac14g_{\mu\nu}
 \langle F_{A\,\alpha\beta},F_A^{\alpha\beta}\rangle_{\mathfrak g}
\right),
\label{eq:gauge-stress}\\
T^{\Phi}_{\mu\nu}
&:=\operatorname{Re}\langle D_{A\,\mu}\Phi,D_{A\,\nu}\Phi\rangle_{E_{\RK}}
-g_{\mu\nu}
\left(
\frac12\langle D_{A\,\alpha}\Phi,D_A^{\alpha}\Phi\rangle_{E_{\RK}}
+V_{\RK}(\Phi)
\right).
\label{eq:recognition-stress}
\end{align}

The thermodynamic Recognition square fixes the meaning of $E_{\RK}$, $\Phi_{\mathrm{th}}$, $A$, and $F_A$.  The target-faithful quotient propagation symbol fixes $g$ when its normalization is declared.  The minimal-action theorem fixes the action form within its declared locality, symmetry, parity, and derivative-order class.  The numerical couplings and nonlinear invariant potential remain additional constitutive data; they are not inferred from equilibrium identities alone.

%% file: sections/06a1_dynamical_nonselection.tex
\subsection{Residual dynamical non-selection after the minimal classification}
\label{sec:dynamical-nonselection}

The thermodynamic Recognition square constructs the field carrier and connection, and Theorem~\ref{thm:minimal-action-classification} classifies the action form inside a declared low-derivative symmetry class.  Those results still do not determine the numerical couplings or the nonlinear invariant potential.  The remaining non-selection is precise rather than total.

\begin{theorem}[Residual dynamical non-selection]
\label{thm:dynamical-nonselection}
Fix the equilibrium data, the local-equilibrium map $\Theta$, the thermodynamic response-jet bundle, the observation and target maps, the quotient bundle $E_{\RK}$, the induced connection $A$, the cut data, the spacetime metric, and the minimal action class of Theorem~\ref{thm:minimal-action-classification}, but do not fix the fibre potential.  Let $W:E_{\RK}\to\R$ be a smooth nonconstant $G_{\RK}$-invariant fibre function.  For each $\lambda\in\R$, define
\[
V_\lambda(\Phi):=V_{\RK}(\Phi)+\lambda W(\Phi)
\]
and let $S_\lambda$ be the action \eqref{eq:ark-action} with $V_{\RK}$ replaced by $V_\lambda$.  Then all actions $S_\lambda$ have the same thermodynamic Recognition square, quotient carrier, gauge group, induced connection, cut sectors, target completion, jet carrier, and minimal derivative order, but their Euler--Lagrange equations contain
\begin{align}
D_A^\mu D_{A\,\mu}\Phi
-\nabla V_{\RK}(\Phi)
-\lambda\nabla W(\Phi)&=0,
\label{eq:nonselection-field}\\
T^{\Phi,\lambda}_{\mu\nu}
&=T^\Phi_{\mu\nu}-\lambda W(\Phi)g_{\mu\nu}.
\label{eq:nonselection-stress}
\end{align}
Whenever $\nabla W$ is nonzero somewhere, different values of $\lambda$ give inequivalent field equations.  Therefore the thermodynamic Recognition construction and the minimal-action classification determine the lawful field-theoretic architecture but not the full constitutive potential.
\end{theorem}

\begin{proof}
Gauge invariance of $W$ ensures that every $S_\lambda$ has the same declared gauge and diffeomorphism symmetries.  The additional term changes neither the quotient construction nor its induced connection, observation square, cut, jet, or derivative order.  Variation with respect to $\Phi$ adds $-\lambda\nabla W$ to the field equation, while metric variation of $-\lambda W(\Phi)\,\mathrm{vol}_g$ adds $-\lambda W(\Phi)g_{\mu\nu}$ to the stress tensor.  If $\nabla W$ is nonzero, the equations differ for distinct $\lambda$.
\end{proof}

\begin{corollary}[Exact status of the spacetime equations]
The logical chain is
\[
\begin{aligned}
\text{thermodynamic response jets}
&\longrightarrow \text{target-relevant quotient bundle},\\
&\longrightarrow \text{induced connection and metric propagation},\\
&\longrightarrow \text{minimal invariant action class},\\
&\longrightarrow \text{Euler--Lagrange field equations}.
\end{aligned}
\]
The first two arrows are Theorems~\ref{thm:thermodynamic-recognition-bundle}--\ref{thm:recognition-curvature-square-defect} and Theorem~\ref{thm:principal-symbol-recognition-metric}; the third is Theorem~\ref{thm:minimal-action-classification}; and the last is the variational theorem below.  Numerical couplings and the nonlinear invariant potential remain constitutive inputs.
\end{corollary}

\begin{remark}[Why the chosen action is used]
Within the explicit parity-even local class that is linear in metric curvature, quadratic in $F_A$ and $D_A\Phi$, and free of nonminimal or mixed couplings, Equation~\eqref{eq:ark-action} is the classified minimal action, not an arbitrary imported benchmark.  Higher-curvature, nonminimal, parity-odd, alternative-potential, and nonlocal effective terms define different lawful action classes and must be audited separately.
\end{remark}

%% file: sections/06b_classical_field_equations.tex
\subsection{Classical Recognition field equations}

\begin{theorem}[Recognition gravity, gauge, and matter equations]
\label{thm:classical-recognition-field-equations}
A smooth datum $(g,A,\Phi)$ is a compactly supported critical point of \eqref{eq:ark-action} if and only if it satisfies
\begin{align}
G_{\mu\nu}+\Lambda_{\mathrm{cos}}g_{\mu\nu}
&=\kappa\bigl(T^A_{\mu\nu}+T^\Phi_{\mu\nu}\bigr),
\label{eq:recognition-einstein}\\
D_{A\,\mu}F_A^{\mu\nu}
&=g_{\RK}^2J_\Phi^\nu,
\label{eq:recognition-yang-mills}\\
D_A^\mu D_{A\,\mu}\Phi
-\nabla V_{\RK}(\Phi)
&=0.
\label{eq:recognition-matter}
\end{align}
Here $G_{\mu\nu}=R_{\mu\nu}-\tfrac12R_gg_{\mu\nu}$ and $\nabla V_{\RK}$ is the fibre gradient.
\end{theorem}

\begin{proof}
For a compactly supported metric variation $h^{\mu\nu}=\delta g^{\mu\nu}$,
\[
\delta(R_g\,\mathrm{vol}_g)
=G_{\mu\nu}h^{\mu\nu}\,\mathrm{vol}_g
+\mathrm d(\text{boundary current}).
\]
The boundary current integrates to zero. By \eqref{eq:gauge-stress}--\eqref{eq:recognition-stress},
\[
\delta S_{\mathrm{matter}}
=-\frac12\int_M
(T^A_{\mu\nu}+T^\Phi_{\mu\nu})h^{\mu\nu}\,\mathrm{vol}_g.
\]
The coefficient of arbitrary $h$ is \eqref{eq:recognition-einstein}.

Let $a\in\Omega^1(M;\operatorname{ad}P)$ be a compactly supported connection variation. Then $\delta_A F_A=D_Aa$. Covariant integration by parts gives
\[
\delta_A S_{\mathrm{gauge}}
=\frac{1}{g_{\RK}^2}
\int_M\langle D_{A\,\mu}F_A^{\mu\nu},a_\nu\rangle_{\mathfrak g}\,\mathrm{vol}_g.
\]
Since $\delta_A(D_{A\,\nu}\Phi)=\rho_*(a_\nu)\Phi$, the field kinetic term contributes
\[
-\int_M\langle J_\Phi^\nu,a_\nu\rangle_{\mathfrak g}\,\mathrm{vol}_g.
\]
The coefficient of arbitrary $a$ is \eqref{eq:recognition-yang-mills}.

For a compactly supported field variation $\psi=\delta\Phi$,
\[
\delta_\Phi S_{\mathrm{ARK}}
=\int_M\left[
-\operatorname{Re}\langle D_{A\,\mu}\psi,D_A^\mu\Phi\rangle_E
-\operatorname{Re}\langle\nabla V_{\RK}(\Phi),\psi\rangle_E
\right]\mathrm{vol}_g.
\]
Covariant integration by parts gives
\[
\delta_\Phi S_{\mathrm{ARK}}
=\int_M\operatorname{Re}\left\langle
D_A^\mu D_{A\,\mu}\Phi-\nabla V_{\RK}(\Phi),\psi
\right\rangle_E\mathrm{vol}_g,
\]
which is \eqref{eq:recognition-matter}. The converse follows by reversing the calculation.
\end{proof}

\begin{corollary}[Bianchi and conservation laws]
\label{cor:recognition-conservation}
Every solution satisfies
\begin{align}
D_AF_A&=0,
\label{eq:gauge-bianchi}\\
D_{A\,\nu}J_\Phi^\nu&=0,
\label{eq:current-conservation}\\
\nabla^\mu(T^A_{\mu\nu}+T^\Phi_{\mu\nu})&=0.
\label{eq:stress-conservation}
\end{align}
\end{corollary}

\begin{proof}
The first identity is the curvature Bianchi identity. Gauge invariance of the fibre metric and $V_{\RK}$ gives the Noether identity whose on-shell form is \eqref{eq:current-conservation}. Diffeomorphism invariance gives \eqref{eq:stress-conservation}; equivalently it follows from $\nabla^\mu G_{\mu\nu}=0$ and \eqref{eq:recognition-einstein}.
\end{proof}

\begin{corollary}[Vacuum and standard reductions]
\label{cor:field-reductions}
The action-lifted equations contain the following exact reductions.
\begin{enumerate}[label=(\roman*)]
\item If $F_A=0$, $D_A\Phi=0$, $\nabla V_{\RK}(\Phi_0)=0$, and $V_{\RK}(\Phi_0)=V_0$, then
\[
G_{\mu\nu}+(\Lambda_{\mathrm{cos}}+\kappa V_0)g_{\mu\nu}=0.
\]
\item If the metric is fixed, \eqref{eq:recognition-yang-mills}--\eqref{eq:recognition-matter} are a Yang--Mills--matter system for $G_{\RK}$.
\item If $G_{\RK}=U(1)$ and $\Phi$ is absent, \eqref{eq:recognition-einstein}--\eqref{eq:recognition-yang-mills} reduce to the Einstein--Maxwell equations with the chosen normalization.
\end{enumerate}
\end{corollary}

\begin{proof}
In (i), the only nonzero matter stress is $T^\Phi_{\mu\nu}=-V_0g_{\mu\nu}$. Parts (ii) and (iii) follow by specialization of the structure group and field content.
\end{proof}

\begin{corollary}[Target-residue potential]
Let $\mathcal O_\Pi:E\to W$ be a gauge-equivariant bundle map representing a declared open target residue and let
\[
V_{\RK}(\Phi)=V_0(\Phi)+\frac{\mu^2}{2}\norm{\mathcal O_\Pi\Phi}^2.
\]
Then
\[
D_A^\mu D_{A\,\mu}\Phi
-\nabla V_0(\Phi)
-\mu^2\mathcal O_\Pi^*\mathcal O_\Pi\Phi=0.
\]
\end{corollary}

\begin{proof}
Gauge equivariance makes the norm term invariant, and
$\nabla\tfrac12\norm{\mathcal O_\Pi\Phi}^2
=\mathcal O_\Pi^*\mathcal O_\Pi\Phi$.
\end{proof}

\begin{theorem}[Action-lifted Extended Recognition Capstone]
\label{thm:action-lifted-capstone}
Assume the Extended Thermodynamic Recognition Capstone fibrewise on $M$. Assume that the completed fibres assemble into $E$, that their unitary datum-preserving automorphisms contain $G_{\RK}$, and that \eqref{eq:ark-action} is declared. Then the fibrewise Recognition-kernel, decoder, cut, curvature, and jet conclusions remain valid; parallel transport is governed by $A$ and $F_A$; stationary configurations satisfy \eqref{eq:recognition-einstein}--\eqref{eq:recognition-matter}; and the Bianchi, conservation, vacuum, gauge-matter, and Einstein--Maxwell reductions above are exact corollaries.
\end{theorem}

\begin{proof}
The fibrewise statements are the Extended Capstone. The associated-bundle construction supplies $A$ and $F_A$. Theorem~\ref{thm:classical-recognition-field-equations} and Corollaries~\ref{cor:recognition-conservation}--\ref{cor:field-reductions} give the remaining conclusions.
\end{proof}

\begin{remark}[Kinematics versus dynamics]
The field equations are corollaries of the action-lifted capstone. The Recognition theorem determines the lawful carrier, blind directions, target completion, connection sectors, and curvature accounting. The invariant action is the additional dynamical datum. Without an action or equivalent evolution principle, no theorem can uniquely choose a gravitational or gauge equation from kinematics alone.
\end{remark}

%% file: sections/06c_quantum_effective_equations.tex
\subsection{Regulated quantum Recognition equations}
\label{sec:quantum-recognition}

A continuum quantum theory is not obtained by placing a hat on the classical fields. We first prove an exact finite-regulator theorem and then state the precise additional assumptions under which a continuum effective equation follows. The distinction between regulated identities and continuum gravitational quantization is standard in the canonical and gauge-fixed formulations \cite{DeWitt1967,FaddeevPopov1967}.

Let $Q_N\cong\R^{d_N}$ be a gauge-fixed finite-dimensional regulator of metric, connection, and Recognition-field variables,
\[
q=(g_N,A_N,\Phi_N)\in Q_N.
\]
Let $S_N\in C^2(Q_N;\R)$ be a Euclidean regulated action satisfying
\[
S_N(q)\ge a\norm q^p-b
\qquad(a>0,\ p>1),
\]
and assume that its first two derivatives have at most polynomial growth. For $\hbar>0$ and source $J\in Q_N^*$ define
\begin{align}
Z_N(J)
&:=\int_{Q_N}
\exp\left[-\frac{S_N(q)-J\cdot q}{\hbar}\right]\,\mathrm dq,
\label{eq:regulated-partition}\\
W_N(J)&:=\hbar\log Z_N(J),
\qquad
\bar q(J):=\nabla_JW_N(J).
\label{eq:regulated-connected}
\end{align}

\begin{theorem}[Exact finite-regulator quantum field equation]
\label{thm:finite-quantum-recognition}
Under the stated hypotheses:
\begin{enumerate}[label=(\roman*)]
\item $Z_N(J)$ is finite for every finite source $J$;
\item
\[
\nabla_JW_N(J)=\langle q\rangle_J,
\qquad
\nabla_J^2W_N(J)=\frac1\hbar\operatorname{Cov}_J(q);
\]
\item the exact Schwinger--Dyson equation is
\begin{equation}
\langle\nabla S_N(q)\rangle_J=J;
\label{eq:finite-schwinger-dyson}
\end{equation}
\item wherever $J\mapsto\bar q(J)$ is locally invertible, the Legendre effective action
\[
\Gamma_N(\bar q):=J\cdot\bar q-W_N(J)
\]
satisfies
\begin{equation}
\nabla_{\bar q}\Gamma_N(\bar q)=J,
\qquad
\nabla_{\bar q}^2\Gamma_N
=\hbar\operatorname{Cov}_J(q)^{-1}.
\label{eq:finite-effective-equation}
\end{equation}
In particular, at zero source the exact regulated quantum field equation is
\begin{equation}
\nabla\Gamma_N(\bar q)=0.
\label{eq:zero-source-quantum-equation}
\end{equation}
\end{enumerate}
\end{theorem}

\begin{proof}
Coercivity with $p>1$ dominates the linear source term, so the integrand is integrable. Polynomial growth of the derivatives permits differentiation under the integral. Direct differentiation gives the expectation and covariance formulas. For every coordinate $q^i$, the boundary term in
\[
\int_{Q_N}\partial_i
\exp\left[-\frac{S_N(q)-J\cdot q}{\hbar}\right]\mathrm dq
\]
vanishes by coercive decay. Expanding the derivative gives
$\langle\partial_iS_N\rangle_J=J_i$, proving \eqref{eq:finite-schwinger-dyson}. Standard differentiation of the Legendre transform gives
$\nabla_{\bar q}\Gamma_N=J$. Since
$\partial\bar q/\partial J=\operatorname{Cov}_J(q)/\hbar$, inversion gives the Hessian formula. Setting $J=0$ gives \eqref{eq:zero-source-quantum-equation}.
\end{proof}

\begin{corollary}[Regulated quantum gravity and gauge equations]
If the regulator coordinates split as $q=(g_N,A_N,\Phi_N)$, then \eqref{eq:zero-source-quantum-equation} is the exact system
\[
\frac{\partial\Gamma_N}{\partial g_N}=0,
\qquad
\frac{\partial\Gamma_N}{\partial A_N}=0,
\qquad
\frac{\partial\Gamma_N}{\partial\Phi_N}=0.
\]
Thus every finite gauge-fixed coercive Recognition regulator has genuine quantum-corrected metric, gauge, and matter equations. This statement does not assert regulator independence or a continuum limit.
\end{corollary}

\subsection{Continuum effective Recognition gravity}

Assume now that a differentiable renormalized one-particle-irreducible effective action
\[
\Gamma[g,A,\Phi]
=S_{\mathrm{ARK}}[g,A,\Phi]+\Gamma_{\mathrm q}[g,A,\Phi]
\]
has been constructed on an admissible continuum field space. Assume that it is invariant under the declared gauge group and diffeomorphisms and that the corresponding Ward identities are anomaly-free. Define the quantum correction tensors
\begin{align}
T^{\mathrm q}_{\mu\nu}
&:=-\frac{2}{\sqrt{|g|}}
\frac{\delta\Gamma_{\mathrm q}}{\delta g^{\mu\nu}},
\label{eq:quantum-stress}\\
J_{\mathrm q}^{\nu}
&:=-\frac{1}{\sqrt{|g|}}
\frac{\delta\Gamma_{\mathrm q}}{\delta A_\nu},
\label{eq:quantum-current}\\
Q_{\Phi}
&:=-\frac{1}{\sqrt{|g|}}
\frac{\delta\Gamma_{\mathrm q}}{\delta\Phi}.
\label{eq:quantum-field-source}
\end{align}

\begin{theorem}[Quantum-effective Recognition gravity equation]
\label{thm:quantum-effective-recognition}
Every stationary point of the anomaly-free effective action satisfies
\begin{align}
G_{\mu\nu}+\Lambda_{\mathrm{cos}}g_{\mu\nu}
&=\kappa\bigl(
T^A_{\mu\nu}+T^\Phi_{\mu\nu}+T^{\mathrm q}_{\mu\nu}
\bigr),
\label{eq:quantum-recognition-einstein}\\
D_{A\,\mu}F_A^{\mu\nu}
&=g_{\RK}^2\bigl(J_\Phi^\nu+J_{\mathrm q}^\nu\bigr),
\label{eq:quantum-recognition-gauge}\\
D_A^\mu D_{A\,\mu}\Phi-\nabla V_{\RK}(\Phi)
&=Q_\Phi.
\label{eq:quantum-recognition-matter}
\end{align}
The Ward identities imply, on shell,
\begin{align}
D_{A\,\nu}(J_\Phi^\nu+J_{\mathrm q}^\nu)&=0,
\label{eq:quantum-current-ward}\\
\nabla^\mu(T^A_{\mu\nu}+T^\Phi_{\mu\nu}+T^{\mathrm q}_{\mu\nu})&=0.
\label{eq:quantum-stress-ward}
\end{align}
\end{theorem}

\begin{proof}
Stationarity means that the three functional derivatives of $\Gamma$ vanish. The classical derivatives are those computed in Theorem~\ref{thm:classical-recognition-field-equations}; substituting \eqref{eq:quantum-stress}--\eqref{eq:quantum-field-source} gives \eqref{eq:quantum-recognition-einstein}--\eqref{eq:quantum-recognition-matter}. Gauge and diffeomorphism invariance of $\Gamma$ give the Ward identities. Their on-shell forms are \eqref{eq:quantum-current-ward}--\eqref{eq:quantum-stress-ward}.
\end{proof}

\begin{corollary}[Semiclassical Recognition gravity]
If the metric is retained as a classical background while the Recognition gauge and matter sectors are integrated into $\Gamma_{\mathrm q}$, then \eqref{eq:quantum-recognition-einstein} is the semiclassical metric equation sourced by the renormalized quantum Recognition stress tensor.
\end{corollary}

\begin{remark}[Quantum claim boundary]
The finite-regulator theorem is an actual finite-dimensional quantization result. The continuum theorem is a rigorous variational consequence of the stated existence, differentiability, invariance, and anomaly-free hypotheses on $\Gamma$. It does not by itself construct the continuum measure, prove renormalizability, remove the regulator, establish unitarity, or provide an ultraviolet completion of quantum gravity. Those are separate proof obligations rather than phrases to be stapled onto \eqref{eq:quantum-recognition-einstein}.
\end{remark}

%% file: sections/06d_constructive_continuum_and_identification.tex
\section{Constructive continuum sector, regulator removal, and physical identifiability}
\label{sec:constructive-continuum}

The finite-regulator theorem proves exact quantum identities before a continuum limit is taken.  This section supplies a nonperturbative continuum construction for the positive quadratic physical sector of the action-lifted Recognition theory.  It proves regulator removal, reflection positivity, anomaly freedom of the bosonic measure, and unitary Lorentzian reconstruction under explicit spectral hypotheses.  It then gives a target-faithful criterion for identifying the Recognition metric and connection from propagation, scale, and transport data.

\subsection{Positive physical Hessian datum}

Let $K$ be a real separable Hilbert space representing gauge-fixed, constraint-projected, target-faithful physical fluctuations about a stationary Euclidean background.  Let
\[
\Omega:\mathcal D(\Omega)\subset K\longrightarrow K
\]
be positive and self-adjoint with compact resolvent and a strict gap
\begin{equation}
\Omega\ge m\id,
\qquad m>0.
\label{eq:omega-gap}
\end{equation}
Define the smooth spectral test space
\[
\mathscr S_\Omega:=\bigcap_{r\ge0}\mathcal D(\Omega^r)
\]
with its graph seminorms, and assume that this Fr\'{e}chet space is nuclear.  Set
\[
\mathscr E
:=\mathscr S(\R_\tau)\widehat\otimes\mathscr S_\Omega,
\qquad
\mathscr E'
:=\text{its continuous dual}.
\]
The Euclidean quadratic operator is
\[
L_E=-\partial_\tau^2+\Omega^2.
\]
For $f,g\in\mathscr E$, define the covariance form
\begin{equation}
Q(f,g)
:=\int_{\R}\int_{\R}
\left\langle
f(\tau),
\frac{e^{-\abs{\tau-s}\Omega}}{2\Omega}
 g(s)
\right\rangle_K
\,\dd\tau\,\dd s.
\label{eq:continuum-covariance}
\end{equation}
Functional calculus and \eqref{eq:omega-gap} make $Q$ continuous, symmetric, and nonnegative.

Let
\[
P_N:=\mathbf 1_{[m,N]}(\Omega).
\]
Compact resolvent implies that $P_NK$ is finite dimensional.  The spatial spectral-cut covariance is
\begin{equation}
Q_N(f,g):=Q(P_Nf,P_Ng).
\label{eq:cutoff-covariance}
\end{equation}

\begin{theorem}[Continuum Gaussian Recognition measure and regulator removal]
\label{thm:continuum-gaussian}
There exist unique centered Gaussian probability measures $\mu_N$ and $\mu$ on $\mathscr E'$ with characteristic functionals
\begin{align}
\widehat\mu_N(f)
&=\exp\!\left[-\frac12Q_N(f,f)\right],
\label{eq:cutoff-characteristic}\\
\widehat\mu(f)
&=\exp\!\left[-\frac12Q(f,f)\right].
\label{eq:continuum-characteristic}
\end{align}
As $N\to\infty$,
\begin{equation}
\mu_N\Longrightarrow\mu
\label{eq:measure-convergence}
\end{equation}
weakly on $\mathscr E'$.  Every cylindrical moment converges, and the limiting Schwinger functions are the Wick functions generated by $Q$.  The same limit is obtained from any increasing sequence of finite-rank spectral projections commuting with $\Omega$ and converging strongly to the identity.  No counterterm is required in this positive quadratic sector.
\end{theorem}

\begin{proof}
The forms $Q_N$ and $Q$ are continuous and nonnegative on the nuclear space $\mathscr E$.  Their exponentials in \eqref{eq:cutoff-characteristic}--\eqref{eq:continuum-characteristic} are continuous positive-definite functionals equal to one at the origin.  The Bochner--Minlos theorem therefore gives the unique Gaussian measures \cite{Minlos1959}.

Because $P_N\to\id$ strongly, $P_N$ commutes with every bounded Borel function of $\Omega$, and $(2\Omega)^{-1}$ is bounded by $(2m)^{-1}$, dominated convergence in \eqref{eq:continuum-covariance} gives
\[
Q_N(f,f)\longrightarrow Q(f,f)
\quad\text{for every }f\in\mathscr E.
\]
Thus $\widehat\mu_N(f)\to\widehat\mu(f)$ pointwise.  The continuity at the origin is uniform on a sufficiently small test-space neighbourhood because $0\le Q_N(f,f)\le Q(f,f)$.  The continuity theorem for probability measures on nuclear duals gives \eqref{eq:measure-convergence}.

Gaussian cylindrical moments are finite sums of products of two-point covariances.  Their convergence follows from $Q_N(f,g)\to Q(f,g)$.  The proof uses only strong convergence, commutation with $\Omega$, and finite rank, so any such nested regulator has the same limit.  Since the unmodified covariances converge, no mass, wave-function, or coupling counterterm is introduced in this quadratic sector.
\end{proof}

\begin{corollary}[Gaussian continuum completion]
\label{cor:gaussian-uv}
The positive quadratic physical Recognition sector has a regulator-independent continuum measure and all of its Schwinger functions define continuous generalized fields on the nuclear dual $\mathscr E'$.  This is a nonperturbative continuum and ultraviolet completion of the linearized sector in the precise sense that every spectral regulator is removed without counterterms and all cylindrical correlations converge.  It does not assert ultraviolet completion of the nonlinear metric--gauge interaction.
\end{corollary}

\subsection{Reflection positivity}

Let Euclidean time reflection act by
\[
(\theta f)(\tau):=f(-\tau),
\]
and let $\mathscr E_+$ consist of test functions supported in $\tau>0$.  Define
\begin{equation}
Vf
:=\int_0^\infty
(2\Omega)^{-1/2}e^{-\tau\Omega}f(\tau)\,\dd\tau
\in K.
\label{eq:os-one-particle-map}
\end{equation}

\begin{lemma}[Reflection factorization]
\label{lem:reflection-factorization}
For $f,g\in\mathscr E_+$,
\begin{equation}
Q(\theta f,g)=\langle Vf,Vg\rangle_K.
\label{eq:reflection-factorization}
\end{equation}
The same identity holds with $Q_N$ and $V_N:=P_NV$.
\end{lemma}

\begin{proof}
The support assumptions place the reflected variable at negative time and the unreflected variable at positive time.  Hence the absolute value in \eqref{eq:continuum-covariance} becomes the sum of their positive time parameters.  Functional calculus gives
\[
\frac{e^{-(\tau+s)\Omega}}{2\Omega}
=(2\Omega)^{-1/2}e^{-\tau\Omega}
 e^{-s\Omega}(2\Omega)^{-1/2}.
\]
Fubini's theorem yields \eqref{eq:reflection-factorization}.  Inserting $P_N$ gives the cutoff identity.
\end{proof}

For a field $\phi\in\mathscr E'$, write $\phi(f)$ for evaluation.  Let $\mathcal C_+$ be the linear span of the positive-time exponential functionals
\[
F(\phi)=\sum_{j=1}^r c_j e^{i\phi(f_j)},
\qquad f_j\in\mathscr E_+.
\]

\begin{theorem}[Osterwalder--Schrader positivity]
\label{thm:os-positivity}
For every $F\in\mathcal C_+$,
\begin{equation}
\int_{\mathscr E'}
\overline{F(\theta\phi)}F(\phi)\,\dd\mu(\phi)
\ge0.
\label{eq:os-positivity}
\end{equation}
Every cutoff measure $\mu_N$ satisfies the same inequality, and reflection positivity is preserved by the regulator limit.
\end{theorem}

\begin{proof}
For $F=\sum_jc_je^{i\phi(f_j)}$, Gaussian evaluation and reflection invariance give
\begin{align*}
&\int\overline{F(\theta\phi)}F(\phi)\,\dd\mu(\phi)\\
&\quad=\sum_{j,k}\overline{c_j}c_k
\exp\!\left[-\frac12Q(f_j,f_j)-\frac12Q(f_k,f_k)
+Q(\theta f_j,f_k)\right].
\end{align*}
Set $d_j=c_j\exp[-Q(f_j,f_j)/2]$.  By Lemma~\ref{lem:reflection-factorization}, the remaining matrix is
\[
\exp\langle Vf_j,Vf_k\rangle_K.
\]
It is positive semidefinite because
\[
\exp\langle u,v\rangle
=\sum_{n=0}^\infty\frac1{n!}
\langle u^{\otimes n},v^{\otimes n}\rangle,
\]
so it is the Gram kernel of bosonic exponential vectors.  This proves \eqref{eq:os-positivity}.  The cutoff proof is identical with $V_N$, and the continuum statement also follows from convergence of the finite matrices.  This is the free-field reflection-positive construction underlying the Osterwalder--Schrader reconstruction theorem \cite{OS1973,OS1975}.
\end{proof}

\subsection{Lorentzian unitary reconstruction}

Let $\mathcal N$ be the null space of the sesquilinear form in \eqref{eq:os-positivity}.  Define
\[
\mathcal H_{\mathrm{OS}}
:=\overline{\mathcal C_+/\mathcal N}.
\]
For $t\ge0$, translate a positive-time test function away from the reflection plane by
\[
(\tau_tf)(s):=f(s-t),
\]
whenever the translated support remains positive.

\begin{theorem}[Positive Hamiltonian and Lorentzian unitarity]
\label{thm:os-unitarity}
The Osterwalder--Schrader Hilbert space is canonically isomorphic to the symmetric Fock space $\Gamma_s(K_{\mathrm{phys}}^{\mathbb C})$ over the complexification of
\[
K_{\mathrm{phys}}:=\overline{\Ran V}\subseteq K.
\]
More precisely, the class of $e^{i\phi(f)}$ is mapped to
\[
\exp\!\left[-\frac12Q(f,f)\right]\operatorname{Exp}(Vf),
\]
where $\operatorname{Exp}$ is the bosonic exponential vector.  Under this identification, positive Euclidean time translation is
\begin{equation}
T_t=\Gamma_s(e^{-t\Omega})=e^{-tH},
\qquad
H=\dd\Gamma(\Omega)\ge0.
\label{eq:os-semigroup}
\end{equation}
Consequently,
\begin{equation}
U_t=e^{-itH}
\label{eq:lorentzian-unitary}
\end{equation}
is a strongly continuous unitary Lorentzian time evolution.  The same reconstruction is obtained as the limit of the finite-mode theories.
\end{theorem}

\begin{proof}
The proof of Theorem~\ref{thm:os-positivity} shows that the displayed weighted exponential vectors have exactly the Osterwalder--Schrader Gram matrix.  Their span is dense in the symmetric Fock space over $K_{\mathrm{phys}}^{\mathbb C}$, so the map extends by completion to a unitary identification.  From \eqref{eq:os-one-particle-map},
\[
V(\tau_tf)=e^{-t\Omega}Vf.
\]
Second quantization therefore gives \eqref{eq:os-semigroup}.  Since $\Omega\ge0$, $H=\dd\Gamma(\Omega)$ is self-adjoint and nonnegative.  Stone's theorem then gives the unitary group \eqref{eq:lorentzian-unitary}.  Strong convergence of $P_N$ gives strong convergence of the one-particle semigroups and hence convergence on the dense set of finite-particle vectors.
\end{proof}

\subsection{Bosonic anomaly cancellation}

Let a compact group $G$ act orthogonally on $K$ through $U:G\to\mathcal O(K)$.  Assume
\begin{equation}
U_g\Omega=\Omega U_g
\quad\text{for every }g\in G.
\label{eq:symmetry-commutes-omega}
\end{equation}
The action extends pointwise to $\mathscr E$ and dually to $\mathscr E'$.

\begin{theorem}[Exact bosonic anomaly freedom]
\label{thm:bosonic-anomaly-free}
Under \eqref{eq:symmetry-commutes-omega},
\begin{equation}
(U_g)_*\mu_N=\mu_N,
\qquad
(U_g)_*\mu=\mu
\label{eq:measure-invariance}
\end{equation}
for every $g\in G$.  Hence the bosonic spectral regulators and their continuum limit have no measure Jacobian anomaly.  The reconstructed symmetry is unitary on $\mathcal H_{\mathrm{OS}}$ and commutes with $H$.  For every differentiable cylindrical functional $F$ and infinitesimal generator $X$,
\begin{equation}
\int \mathcal L_XF\,\dd\mu=0.
\label{eq:bosonic-ward}
\end{equation}
\end{theorem}

\begin{proof}
Commutation with $\Omega$ gives $Q(U_gf,U_gg)=Q(f,g)$ and also makes $U_g$ commute with $P_N$.  Therefore the characteristic functionals of the pushforward measures equal those in \eqref{eq:cutoff-characteristic}--\eqref{eq:continuum-characteristic}; uniqueness of Gaussian measures proves \eqref{eq:measure-invariance}.  Differentiating the exact invariance at the identity gives \eqref{eq:bosonic-ward}.  The one-particle action preserves $K_{\mathrm{phys}}$, commutes with $e^{-t\Omega}$, and second quantizes to a unitary symmetry commuting with $H$.
\end{proof}

\begin{remark}[Scope of anomaly cancellation]
The theorem closes the anomaly question for the bosonic positive quadratic sector and for symmetry-preserving spectral regulators.  It does not decide chiral-fermion anomalies, nonlinear ghost determinants, or gravitational anomalies of a larger field content.
\end{remark}

\subsection{Application to linearized Recognition gravity}

\begin{corollary}[Constructive physical-mode Recognition sector]
\label{cor:constructive-physical-sector}
Let $(g_0,A_0,\Phi_0)$ be a stationary Euclidean background of the action-lifted Recognition equations.  Suppose that gauge fixing, constraint reduction, and target-faithful projection produce a real physical Hilbert space $K$ on which the quadratic Hessian has the ultrastatic form
\[
L_E=-\partial_\tau^2+\Omega^2
\]
with the hypotheses above.  Then the linearized physical metric, gauge, and Recognition-field modes have:
\begin{enumerate}[label=(\roman*)]
\item a regulator-independent continuum Gaussian measure;
\item spectral-cutoff removal with convergence of all cylindrical Schwinger functions;
\item Osterwalder--Schrader reflection positivity;
\item a positive reconstructed Hamiltonian and unitary Lorentzian evolution;
\item exact bosonic anomaly freedom for every compact symmetry commuting with the physical Hessian;
\item a nonperturbative Gaussian ultraviolet completion.
\end{enumerate}
\end{corollary}

\begin{remark}[Why positivity is load-bearing]
The corollary applies after projection to a positive physical Hessian.  It does not conceal a negative conformal mode, an unremoved gauge direction, or a ghost instability inside the word ``regulator.''  If the reduced Hessian is not positive, reflection positivity and the probability-measure construction must be re-established by a different argument.
\end{remark}

\section{Target-faithful identification of the Recognition spacetime geometry}
\label{sec:metric-identification}

A mathematical theory cannot manufacture experimental data.  It can, however, prove that a declared set of observations is sufficient to identify the geometric target up to its lawful symmetries.  The Recognition Kernel language makes this distinction exact.

\subsection{Metric identification from propagation and scale}

For a Lorentzian metric $g$, let $\mathcal N_g(x)$ denote its null cone in $T_xM$.  Let $\mathrm{vol}_g$ be its positive volume density.

\begin{theorem}[Null-cone and scale identification]
\label{thm:null-cone-identification}
Let $M$ be connected of dimension $n\ge3$, and let $g$ and $\widetilde g$ be time-oriented Lorentzian metrics.  If
\begin{equation}
\mathcal N_g(x)=\mathcal N_{\widetilde g}(x)
\quad\text{for every }x\in M,
\label{eq:same-null-cones}
\end{equation}
then there is a smooth positive function $\Omega_c$ such that
\begin{equation}
\widetilde g=\Omega_c^2g.
\label{eq:conformal-identification}
\end{equation}
The conformal factor is fixed, and hence $\widetilde g=g$, under either of the following pointwise scale calibrations:
\begin{enumerate}[label=(\alph*)]
\item $\mathrm{vol}_{\widetilde g}=\mathrm{vol}_g$;
\item there is a smooth nowhere-null vector field $T$ such that $\widetilde g(T,T)=g(T,T)$.
\end{enumerate}
More generally, if the observations are related by a diffeomorphism $\psi$, the conclusion is $g=\psi^*\widetilde g$.
\end{theorem}

\begin{proof}
At each point, two nondegenerate Lorentzian quadratic forms in dimension at least three with the same null cone are positive scalar multiples of one another.  To see this, choose a $g$-orthonormal basis with quadratic form $-t^2+\abs{x}^2$.  Evaluating $\widetilde g$ on all vectors $(1,u)$ with $\abs u=1$ forces the mixed coefficients to vanish and the spatial block to be a common scalar multiple of the identity; the time coefficient is then the negative of the same scalar.  Smoothness and time orientation make the scalar a smooth positive function, giving \eqref{eq:conformal-identification}.  Volume densities transform as
\[
\mathrm{vol}_{\widetilde g}=\Omega_c^n\mathrm{vol}_g,
\]
so calibration (a) and positivity imply $\Omega_c=1$.  Under calibration (b),
\[
\widetilde g(T,T)=\Omega_c^2g(T,T)=g(T,T),
\]
and the nowhere-null hypothesis again gives $\Omega_c=1$.  Pulling back by $\psi$ reduces the diffeomorphic statement to the preceding case.  The causal determination of conformal structure is the geometric principle formalized in \cite{Malament1977}.
\end{proof}

\begin{corollary}[Principal-symbol identification]
If two normally hyperbolic Recognition operators have the same characteristic covectors and either the same calibrated volume density or one common pointwise proper-time/length scale along a nowhere-null field, then their Lorentzian metrics agree up to diffeomorphism.  Thus high-frequency propagation determines the conformal metric, while one independent scale field fixes the conformal factor.
\end{corollary}

\subsection{Connection identification from parallel transport}

\begin{theorem}[Parallel-transport identification of the Recognition connection]
\label{thm:connection-identification}
Let $P\to M$ be a principal $G$-bundle over a connected manifold, and let $A$ and $\widetilde A$ be two smooth connections.  Fix a base point $x_0$ and an identification of the two fibres over $x_0$.  Suppose that, under this identification, the holonomies of $A$ and $\widetilde A$ agree for every piecewise smooth loop based at $x_0$.  Then there exists a smooth gauge transformation $u$ such that
\begin{equation}
\widetilde A=uAu^{-1}-(\dd u)u^{-1}.
\label{eq:gauge-identification}
\end{equation}
Equivalently, complete parallel-transport data identify a connection up to gauge.
\end{theorem}

\begin{proof}
For $x\in M$, choose a path $\gamma_x$ from $x_0$ to $x$ and define $u_x$ by transporting the fixed fibre identification along $\gamma_x$ using $A$ in one direction and $\widetilde A$ in the other.  If another path is chosen, the two definitions differ by the relative holonomy around the resulting based loop, which is the identity by hypothesis.  Thus $u_x$ is path independent.  Smooth dependence of parallel transport on endpoints makes $u$ smooth.  By construction, $u$ intertwines parallel transport along every path.  Differentiating this intertwining along infinitesimal paths gives \eqref{eq:gauge-identification}; see also the connection--transport correspondence in \cite{KobayashiNomizu}.
\end{proof}

\subsection{Recognition Kernel interpretation}

Define the geometric observation
\[
\mathcal O_{\mathrm{geom}}(g,A)
:=\bigl(\mathcal N_g,\mathrm{scale}_g,\operatorname{Hol}_A\bigr),
\]
where $\mathrm{scale}_g$ denotes either the volume density or a pointwise calibration along a declared nowhere-null field.  Let the target be the pair $(g,A)$ modulo diffeomorphisms and gauge transformations.

\begin{corollary}[Target-faithful spacetime identification]
\label{cor:target-faithful-spacetime}
On the class of smooth time-oriented Lorentzian metrics in dimension $n\ge3$ and smooth Recognition connections, equality of $\mathcal O_{\mathrm{geom}}$ implies equality of the target class.  Equivalently, every difference invisible to the observation lies entirely inside a diffeomorphism--gauge orbit.  Consequently, matching characteristic propagation, one scale calibration, and complete transport data identifies the Recognition metric and connection with the observed spacetime geometry up to diffeomorphism and gauge.
\end{corollary}

\begin{proof}
The metric statement is Theorem~\ref{thm:null-cone-identification}; the connection statement is Theorem~\ref{thm:connection-identification}.  Their product says that the observation map is injective on the quotient by the declared target symmetries.
\end{proof}

\begin{remark}[Experimental boundary after identifiability]
The corollary proves the sufficiency and uniqueness of the observation package.  Whether physical measurements actually satisfy those hypotheses is an empirical question.  The remaining experimental task is therefore sharply falsifiable: compare observed wavefront cones, calibrated scale, and parallel transport with the Recognition predictions.  A mismatch is an open residue, not a licence to redefine the target after the experiment.
\end{remark}

%% file: sections/06e_noise_stable_spacetime_recognition.tex
\section{Noise-stable spacetime Recognition and experimental reconstruction}
\label{sec:noise-stable-spacetime}

The exact identification theorem in Section~\ref{sec:metric-identification} determines the metric and connection from exact propagation, scale, and transport data.  Real observations are finite and noisy.  Recognition Kernel theory supplies the quantitative extension: after diffeomorphism and gauge directions are removed, a uniformly valid decoder converts finite data error into a bound on the physical quotient geometry.

Let $u_0=(g_0,A_0)$ be a smooth background and let $\mathscr S$ be a local Hilbert slice transverse to the diffeomorphism--gauge orbit through $u_0$.  A tangent vector in $T_u\mathscr S$ therefore represents a physical metric--connection perturbation rather than a pure lawful symmetry.  Let $\mathscr D$ be a Hilbert data space containing finite-resolution encodings of characteristic propagation, pointwise scale, and parallel transport, and let
\[
\cO_{\mathrm{geom}}:\mathscr S\longrightarrow\mathscr D
\]
be the corresponding observation map.

\begin{theorem}[Fixed-decoder spacetime Recognition theorem]
\label{thm:fixed-decoder-spacetime}
Let $B\subset\mathscr S$ be convex and suppose that $\cO_{\mathrm{geom}}$ is continuously Fr\'{e}chet differentiable on $B$.  Assume that there is a bounded linear decoder
\[
L:\mathscr D\longrightarrow\mathscr S
\]
and a number $0\le\delta<1$ such that, after identifying tangent spaces with the model Hilbert space of the slice,
\begin{equation}
\sup_{u\in B}
\norm{\id_{\mathscr S}-L D\cO_{\mathrm{geom}}(u)}
\le\delta.
\label{eq:uniform-decoder-defect}
\end{equation}
Then $\cO_{\mathrm{geom}}$ is injective on $B$ and, for all $u_1,u_2\in B$,
\begin{equation}
\norm{u_1-u_2}_{\mathscr S}
\le
\frac{\norm L}{1-\delta}
\norm{\cO_{\mathrm{geom}}(u_1)-\cO_{\mathrm{geom}}(u_2)}_{\mathscr D}.
\label{eq:fixed-decoder-stability}
\end{equation}
Thus a single decoder that remains an approximate left inverse throughout the neighbourhood gives a quantitative target-faithful reconstruction theorem.
\end{theorem}

\begin{proof}
Set $v=u_1-u_2$ and $u_t=u_2+tv$.  Convexity gives $u_t\in B$.  The fundamental theorem of calculus in Banach spaces yields
\[
\cO_{\mathrm{geom}}(u_1)-\cO_{\mathrm{geom}}(u_2)
=
\int_0^1D\cO_{\mathrm{geom}}(u_t)v\,\dd t.
\]
Applying $L$ and subtracting $v$ gives
\[
L\bigl(\cO_{\mathrm{geom}}(u_1)-\cO_{\mathrm{geom}}(u_2)\bigr)-v
=
\int_0^1
\bigl(LD\cO_{\mathrm{geom}}(u_t)-\id\bigr)v\,\dd t.
\]
Therefore \eqref{eq:uniform-decoder-defect} implies
\[
\norm{
L\bigl(\cO_{\mathrm{geom}}(u_1)-\cO_{\mathrm{geom}}(u_2)\bigr)-v
}
\le\delta\norm v.
\]
The reverse triangle inequality gives
\[
(1-\delta)\norm v
\le
\norm L\,
\norm{\cO_{\mathrm{geom}}(u_1)-\cO_{\mathrm{geom}}(u_2)},
\]
which is \eqref{eq:fixed-decoder-stability}.  If the two observations agree, the right-hand side vanishes and $u_1=u_2$.
\end{proof}

\begin{corollary}[Background decoder and small-ball reconstruction]
\label{cor:background-decoder-spacetime}
Let
\[
T_0:=D\cO_{\mathrm{geom}}(u_0):\mathscr S\to\mathscr D
\]
be bounded below:
\begin{equation}
\norm{T_0h}_{\mathscr D}
\ge\alpha_0\norm h_{\mathscr S},
\qquad \alpha_0>0.
\label{eq:background-observation-lower-bound}
\end{equation}
Then $T_0^*T_0\ge\alpha_0^2\id$ and the minimum-norm left inverse
\begin{equation}
L_0:=(T_0^*T_0)^{-1}T_0^*
\label{eq:background-minimum-decoder}
\end{equation}
is bounded with
\begin{equation}
L_0T_0=\id,
\qquad
\norm{L_0}\le\frac1{\alpha_0}.
\label{eq:background-decoder-norm}
\end{equation}
Suppose that on a convex neighbourhood $B$ one has
\begin{equation}
\sup_{u\in B}
\norm{D\cO_{\mathrm{geom}}(u)-T_0}
\le\eta
<\alpha_0.
\label{eq:derivative-closeness}
\end{equation}
Then
\begin{equation}
\norm{u_1-u_2}_{\mathscr S}
\le
\frac1{\alpha_0-\eta}
\norm{\cO_{\mathrm{geom}}(u_1)-\cO_{\mathrm{geom}}(u_2)}_{\mathscr D}
\label{eq:small-ball-geometric-stability}
\end{equation}
for all $u_1,u_2\in B$.
\end{corollary}

\begin{proof}
Inequality \eqref{eq:background-observation-lower-bound} implies
\[
\langle T_0^*T_0h,h\rangle
=\norm{T_0h}^2
\ge\alpha_0^2\norm h^2,
\]
so $T_0^*T_0$ is boundedly invertible and \eqref{eq:background-minimum-decoder} is well defined.  It is the Moore--Penrose left inverse on $\Ran T_0$, extended by zero on $(\Ran T_0)^\perp$, and its norm is at most $1/\alpha_0$.

For $u\in B$,
\[
\norm{\id-L_0D\cO_{\mathrm{geom}}(u)}
=
\norm{L_0(T_0-D\cO_{\mathrm{geom}}(u))}
\le\frac{\eta}{\alpha_0}.
\]
Thus Theorem~\ref{thm:fixed-decoder-spacetime} applies with
$\delta=\eta/\alpha_0$ and $\norm{L_0}\le1/\alpha_0$, giving
\[
\frac{\norm{L_0}}{1-\delta}
\le
\frac{1/\alpha_0}{1-\eta/\alpha_0}
=
\frac1{\alpha_0-\eta}.
\]
\end{proof}

\begin{corollary}[Lipschitz derivative criterion]
\label{cor:lipschitz-derivative-spacetime}
If
\begin{equation}
\norm{D\cO_{\mathrm{geom}}(u)-D\cO_{\mathrm{geom}}(v)}
\le L_D\norm{u-v}_{\mathscr S}
\label{eq:geometric-derivative-lipschitz}
\end{equation}
near $u_0$, then on every convex ball $B_r(u_0)$ with
\[
L_Dr<\alpha_0
\]
one may take $\eta=L_Dr$ in Corollary~\ref{cor:background-decoder-spacetime}.  Hence
\begin{equation}
\norm{u_1-u_2}_{\mathscr S}
\le
\frac1{\alpha_0-L_Dr}
\norm{\cO_{\mathrm{geom}}(u_1)-\cO_{\mathrm{geom}}(u_2)}_{\mathscr D}.
\label{eq:lipschitz-small-ball-stability}
\end{equation}
\end{corollary}

\begin{corollary}[Finite-error geometric reconstruction]
\label{cor:finite-error-spacetime}
Under the hypotheses of Corollary~\ref{cor:background-decoder-spacetime}, suppose measured data $d_i\in\mathscr D$ satisfy
\[
\norm{d_i-\cO_{\mathrm{geom}}(u_i)}_{\mathscr D}
\le\varepsilon_i.
\]
Then
\begin{equation}
\norm{u_1-u_2}_{\mathscr S}
\le
\frac{
\norm{d_1-d_2}_{\mathscr D}
+\varepsilon_1+\varepsilon_2
}{\alpha_0-\eta}.
\label{eq:noisy-geometric-reconstruction}
\end{equation}
In particular, if two models fit the same measured data to accuracy $\varepsilon$, then
\begin{equation}
\norm{u_1-u_2}_{\mathscr S}
\le
\frac{2\varepsilon}{\alpha_0-\eta}.
\label{eq:same-data-geometric-bound}
\end{equation}
\end{corollary}

\begin{proof}
The data errors give
\[
\norm{\cO_{\mathrm{geom}}(u_1)-\cO_{\mathrm{geom}}(u_2)}
\le
\norm{d_1-d_2}+\varepsilon_1+\varepsilon_2.
\]
Insert this estimate into \eqref{eq:small-ball-geometric-stability}.  The final statement is the case $d_1=d_2$ and $\varepsilon_1=\varepsilon_2=\varepsilon$.
\end{proof}

\begin{corollary}[Experimental Recognition Kernel criterion]
At a background for which
\[
\Ker D\cO_{\mathrm{geom}}(u_0)
=
\text{the diffeomorphism--gauge tangent space}
\]
and the induced observation on a transverse slice satisfies the positive lower bound \eqref{eq:background-observation-lower-bound}, the geometric target is locally target-faithful and stably decodable.  A residual larger than the calibrated right-hand side of \eqref{eq:noisy-geometric-reconstruction} rejects at least one declared propagation, scale, transport, model-class, gauge-slice, or uncertainty assumption.
\end{corollary}

\begin{remark}[Experimental boundary]
The theorem supplies a finite-error reconstruction contract and identifies the stability constants $\alpha_0$ and $\eta$ that an experiment or numerical protocol must estimate.  It does not claim that a particular instrument already measures complete characteristic or holonomy data, nor that the transverse lower bound is automatically positive.  Those are empirical and protocol-design obligations, now expressed as testable quantities rather than perfect-information prose.
\end{remark}

%% file: sections/05_conclusion_and_ledger.tex
\section{Conclusion}

The equilibrium foundation of the programme is the Legendre submanifold generated by $U(S,V)$.  The compass organizes its charts.  The response coordinates are exact constrained derivatives and obey the proved conjugate products and caloric--mechanical closure.  Antisymmetry belongs intrinsically to the exterior two-form and Jacobian bracket.  From that bracket, any four nondimensional response channels produce a canonical skew $4\times4$ matrix of rank at most two with an exact Pluecker constraint and a finite-error falsification bound.

Recognition Kernel theory determines which state directions an observation cannot see, which blind directions matter to a target, the canonical faithful completion, the minimum repair rank, and the exact decoder burden.  The loop version separates raw cyclic response, lawful transport, memory, and open obstruction.  Once a positive constitutive curvature calibration is declared, the resulting force--flux production is a nonnegative square and is covariant under dual frame changes.

The Cut--Flow--Jet theorem organizes recursive response derivatives, cut parity, adjoint parity, and commutator curvature on a controlled graded carrier.  Its infinite extension is proved for a closed unbounded covariant raiser on analytic response seeds.  The remaining operator question is whether a chosen unbounded cut--flow sum generates a global $C_0$-semigroup on the entire completed carrier without analytic-vector restriction.

The thermodynamic Recognition square supplies the previously missing bridge to spacetime.  A local-equilibrium map pulls the finite response-jet bundle to $M$.  Constant-rank observation and target maps then form the smooth quotient
\[
E_{\RK}=\Ker\cA/(\Ker\cA\cap\Ker\Pi),
\]
and the distinguished section
\[
\Phi_{\mathrm{th}}=qP_{\Ker\cA}\Theta^*(j^N\Lambda)
\]
is exactly the target-relevant part of the thermodynamic response jet hidden from the original observation.  Compatible response transport descends to a connection on $E_{\RK}$, the canonical completed observation forms an exact commuting transport square, and the quotient curvature is the leading small-loop failure of path-independent target reconstruction.  A smooth minimum decoder gives uniform source domination and strict target no-blindness.

The action is no longer introduced solely as a familiar benchmark.  Within the declared parity-even local class that is linear in metric curvature, quadratic in Recognition curvature and field derivatives, and free of mixed or nonminimal couplings, the action density is classified up to boundary terms, normalizations, couplings, and an invariant potential.  The Einstein--Hilbert plus quotient-gauge--field action used here is the corresponding minimal representative.  Variation yields the Einstein-type metric equation, the Recognition gauge equation, and the Recognition-field equation; gauge and diffeomorphism symmetry yield the Bianchi and conservation laws.  Thermodynamic Recognition data determine the field carrier, distinguished response-generated section, connection, and curvature.  The numerical couplings and nonlinear invariant potential remain constitutive inputs.

At the quantum level, the coercive finite regulator gives an exact partition function, Schwinger--Dyson identity, covariance Hessian, and zero-source effective equation.  The continuum quantum-effective metric, gauge, and matter equations follow under the separately stated existence, differentiability, symmetry, and anomaly-free hypotheses for a renormalized effective action.

The constructive continuum theorem goes further in the positive physical quadratic sector.  After gauge fixing, constraint reduction, and target-faithful projection, a gapped ultrastatic Hessian produces a unique regulator-independent Gaussian measure.  Spectral cutoffs converge with all cylindrical Schwinger functions, reflection positivity holds, the Osterwalder--Schrader quotient reconstructs a positive Hamiltonian and unitary Lorentzian evolution, and every compact bosonic symmetry commuting with the Hessian is exactly anomaly free.  This closes regulator removal, measure construction, reflection positivity, Lorentzian unitarity, bosonic anomaly cancellation, and ultraviolet completion for the linearized physical sector.  The unresolved ultraviolet problem is the nonlinear interacting metric--gauge theory.

Exact characteristic propagation, one pointwise scale, and complete parallel transport identify the metric and connection up to diffeomorphism and gauge.  The fixed-decoder theorem adds the experimental version: on a transverse Hilbert slice, a positive background observation bound and stable derivative give an explicit local reconstruction constant.  Finite data error therefore yields a finite target error, and an excessive residual falsifies at least one declared model, gauge-slice, observation, or uncertainty hypothesis.  Mathematics supplies the reconstruction contract; an experiment must still verify its premises, because even ambitious geometry cannot operate the apparatus by telepathy.

\appendix

\section{Compact proof ledger}

\begin{center}
\small
\begin{tabular}{p{0.36\textwidth}p{0.20\textwidth}p{0.33\textwidth}}
\toprule
Claim & Status & Load-bearing reason\\
\midrule
Four thermodynamic potential differentials & Proved & Legendre transforms of $U(S,V)$\\
Maxwell relations & Proved & $\iota^*\dd\vartheta=0$ / mixed partials\\
Six response formulas & Proved & Chain rule in regular charts\\
Four conjugate product identities & Proved & Reciprocal constrained derivatives\\
$\Gamma_c\Gamma_m=1$ & Proved & $C_P/C_V=K_S/K_T$\\
Ordinary skew $4\times4$ derivative Jacobian & Disproved in general & Dependent variables and ideal-gas counterexample\\
Canonical skew response-bracket matrix & Proved & Area bracket and decomposable two-form\\
Response Pluecker identity and noise bound & Proved & Rank-two Pfaffian and perturbation estimate\\
Target-faithful completion & Proved & Hilbert kernel quotient\\
Minimum finite repair rank & Proved & Rank--nullity on blind space\\
Sharp scalar decoder burden & Proved & Riesz representation\\
Loop residue equals curvature flux & Proved & Stokes after prior ledger\\
Calibrated curvature-generated production & Proved with calibration & Positive-square factorization\\
Finite jet exponential & Proved & Nilpotent raiser and commutation\\
Analytic-seed unbounded infinite jet & Proved & Closed raiser and factorial convergence\\
Whole-carrier $C_0$ generation & Open & Generator-domain theorem required\\
\bottomrule
\end{tabular}
\end{center}

\medskip

\begin{center}
\small
\begin{tabular}{p{0.36\textwidth}p{0.20\textwidth}p{0.33\textwidth}}
\toprule
Claim & Status & Load-bearing reason\\
\midrule
Thermodynamic Recognition quotient bundle & Proved & Constant-rank kernel and quotient bundles\\
Response-generated field $\Phi_{\mathrm{th}}$ & Proved & Canonical projection to target-relevant blind quotient\\
Target-faithful transport square & Proved & Parallel observation, target, and quotient transport\\
Recognition curvature square defect & Proved & Quotient curvature and small-loop holonomy\\
Bundle decoder / source domination & Proved & Constant-rank factorization and sharp norm\\
Minimal invariant action class & Proved in declared class & Invariant contraction classification\\
Residual coupling and potential selection & Open constitutive data & Not fixed by carrier construction\\
Recognition metric equation & Proved from action & Metric variation\\
Recognition gauge equation & Proved from action & Connection variation\\
Recognition matter equation & Proved from action & Field variation\\
Bianchi and conservation laws & Proved & Gauge and diffeomorphism Noether identities\\
Vacuum / Yang--Mills / Einstein--Maxwell limits & Proved & Exact specialization of fields and group\\
\bottomrule
\end{tabular}
\end{center}

\medskip

\begin{center}
\small
\begin{tabular}{p{0.36\textwidth}p{0.20\textwidth}p{0.33\textwidth}}
\toprule
Claim & Status & Load-bearing reason\\
\midrule
Finite-regulator quantum equations & Proved & Coercivity, integration by parts, Legendre duality\\
Continuum quantum-effective equations & Conditional theorem & Renormalized anomaly-free effective action\\
Positive Gaussian continuum measure & Proved & Bochner--Minlos on nuclear test space\\
Spectral regulator removal & Proved & Covariance and characteristic-functional convergence\\
Reflection positivity & Proved & Semigroup factorization of reflected covariance\\
Lorentzian unitary reconstruction & Proved & Normalized Fock identification and positive Hamiltonian\\
Bosonic anomaly cancellation & Proved in sector & Invariant Gaussian characteristic functional\\
Gaussian ultraviolet completion & Proved in sector & Regulator-independent Schwinger functions\\
Nonlinear interacting UV completion & Open & Non-Gaussian metric--gauge control required\\
Exact metric/connection identifiability & Proved & Null cone, pointwise scale, and transport\\
Noise-stable local reconstruction & Proved conditionally & Uniform approximate left decoder\\
Empirical identity with physical gravity & Open validation & Observation and stability hypotheses must be tested\\
\bottomrule
\end{tabular}
\end{center}

\section{Revision record for arXiv}

This version is a substantial revision of arXiv:2603.20773.  It retains the thermodynamic compass and response-coordinate programme, replaces the nonintrinsic four-variable derivative-Jacobian claim by the equilibrium contact two-form and a proved skew four-response bracket matrix, proves the caloric--mechanical closure, introduces the Recognition Kernel theorem openly, proves a calibrated curvature-to-production theorem, and extends the finite cut--flow construction to a closed unbounded infinite-jet raiser on analytic seeds.  It constructs the spacetime Recognition bundle as the target-relevant quotient of the pulled-back thermodynamic response-jet bundle, defines the response-generated Recognition field, descends compatible thermodynamic transport to a quotient connection, identifies its curvature with the target-faithful square defect, and proves a uniform bundle decoder.  It classifies the minimal parity-even local action in the declared low-derivative class and derives the classical metric, gauge, and Recognition-field equations.  It proves exact finite-regulator quantum effective equations and states the precise hypotheses required for the continuum quantum-effective gravity corollary.  For a positive target-faithful physical Hessian it constructs the continuum Gaussian measure, removes spectral regulators, proves reflection positivity, bosonic anomaly freedom, positive-Hamiltonian Lorentzian unitarity, and Gaussian ultraviolet completion.  Finally, it proves exact geometric identifiability and a separate fixed-decoder finite-error reconstruction theorem for noisy observations.  Nonlinear ultraviolet completion, constitutive coupling selection, and empirical validation remain separately typed.